\documentclass[a4paper,10pt]{revtex4-1}

\usepackage{amsmath}
\usepackage{amssymb}
\usepackage{amsthm}
\usepackage{array}
\usepackage{stmaryrd}
\usepackage{graphicx}
\usepackage{dcolumn}
\usepackage{bm}
\usepackage{mathtools}
\usepackage{lipsum}
\usepackage{slashbox}
\usepackage{fontenc}
\usepackage{epstopdf}

\usepackage{algorithm}
\usepackage{algpseudocode}
\newtheorem{proposition}{Proposition}[section]

\newtheorem{rem}{Remark    }[section]

\newcommand{\ket}[1]{{\left\vert{#1}\right\rangle}}
\newcommand{\Syst}[2]{\left\{\begin{array}{ccccc} #1\\ #2 \end{array}\right.}

\def\CC{\mathbb{C}}
\def\PP{\mathbb{P}}
\def\calH{\mathcal{H}}
\def\calO{\mathcal{O}}

\usepackage[all]{xy}
\xyoption{matrix}
\xyoption{frame}
\xyoption{arrow}
\xyoption{arc}

\usepackage{ifpdf}
\ifpdf
\else
\PackageWarningNoLine{Qcircuit}{Qcircuit is loading in Postscript mode.  The Xy-pic options ps and dvips will be loaded.  If you wish to use other Postscript drivers for Xy-pic, you must modify the code in Qcircuit.tex}
\xyoption{ps}
\xyoption{dvips}
\fi

\entrymodifiers={!C\entrybox}

\newcommand{\qw}[1][-1]{\ar @{-} [0,#1]}
\newcommand{\qwx}[1][-1]{\ar @{-} [#1,0]}

\newcommand{\gate}[1]{*+<.6em>{#1} \POS ="i","i"+UR;"i"+UL **\dir{-};"i"+DL **\dir{-};"i"+DR **\dir{-};"i"+UR **\dir{-},"i" \qw}
\newcommand{\meter}{*=<1.8em,1.4em>{\xy ="j","j"-<.778em,.322em>;{"j"+<.778em,-.322em> \ellipse ur,_{}},"j"-<0em,.4em>;p+<.5em,.9em> **\dir{-},"j"+<2.2em,2.2em>*{},"j"-<2.2em,2.2em>*{} \endxy} \POS ="i","i"+UR;"i"+UL **\dir{-};"i"+DL **\dir{-};"i"+DR **\dir{-};"i"+UR **\dir{-},"i" \qw}

\newcommand{\control}{*!<0em,.025em>-=-<.2em>{\bullet}}

\newcommand{\ctrl}[1]{\control \qwx[#1] \qw}

\newcommand{\multigate}[2]{*+<1em,.9em>{\hphantom{#2}} \POS [0,0]="i",[0,0].[#1,0]="e",!C *{#2},"e"+UR;"e"+UL **\dir{-};"e"+DL **\dir{-};"e"+DR **\dir{-};"e"+UR **\dir{-},"i" \qw}
\newcommand{\ghost}[1]{*+<1em,.9em>{\hphantom{#1}} \qw}

\newcommand{\lstick}[1]{*!R!<.5em,0em>=<0em>{#1}}
\newcommand{\ustick}[1]{*!D!<0em,-.5em>=<0em>{#1}}

\newcommand{\Qcircuit}{\xymatrix @*=<0em>}

\def\SLOCC{\text{SLOCC}}

\DeclarePairedDelimiter\ceil{\lceil}{\rceil}
\DeclarePairedDelimiter\floor{\lfloor}{\rfloor}

\begin{document}


\title{Quantum Entanglement involved in Grover's and Shor's algorithms: \\the four-qubit case}

\author{Hamza Jaffali}%
 \email{hamza.jaffali@utbm.fr, FEMTO-ST/UTBM, Universit\'e de Bourgogne Franche-Comt\'e, 90010 Belfort Cedex, France.}
\author{Fr\'ed\'eric Holweck}
\email{frederic.holweck@utbm.fr, ICB/UTBM, Universit\'e de Bourgogne Franche-Comt\'e, 90010 Belfort Cedex, France
}%

%
%

\date{\today}

\begin{abstract}
In this paper, we study the nature of entanglement in quantum Grover's and Shor's algorithms. So far, the authors who have been interested in this problem have approached the question quantitatively by introducing entanglement measures (numerical ones most of the time). One can ask a different question: what about a qualitative measure of entanglement ? In other words, we try to find what are the different entanglement SLOCC classes that can be generated by these two algorithms. We treat in this article the case of pure four-qubit systems.
\begin{description}
\item[Keywords] Entanglement, Four-qubits systems, Gover's algorithmn, Shor's algorithm, Quantum Fourier Transform, Periodic States.
\item[PACS] 02.40.-k, 03.65.Fd, 03.67.-a, 03.67.Lx, 03.65.Ud
\end{description}
\end{abstract}

\maketitle


\section{Introduction}\label{intro}

Nowadays, Quantum Computation and Quantum Information Theory are considered as valuable candidates for the future of computer science and information processing. The existence in the litterature of quantum algorithms, quantum communication protocols and quantum cryptographic schemes, that outclass their classical counterparts, leads naturally to ask the question of what makes quantum computation so efficient.

~

One of the possible answers is to look at quantum entanglement. Quantum entanglement is considered as one of the most important ressource in quantum computation and quantum information processing. It has been proved that for quantum algorithms, entanglement is necessary to provide speed-up \cite{Jozsa,Jozsa2}. The problem of understanding  entanglement have interested many scientists over the last 80 years with a scientific production going from theoretical interpretation to experimental manipulation of entanglement. Today with the recent development of quantum technologies, it is still an open problem to understand how quantum entanglement appears and evolves in quantum computation. 

~

So far, the authors who have been interested in this problem have most of the time approached the question quantitatively by using entanglement measures \cite{Haddadi} (numerical ones most of the time). Majorization theory was also applied to study quantum algorithms \cite{Latorre,Orus,Orus2}, showing that there is a majorization principle underlying the way quantum algorithms operate (even for every step of the Quantum Fourier Transform). 

~

In this paper we investigate quantum entanglement, using tools from Algebraic Geometry and Invariant Theory, to give a qualitative description of the quantum states involved in two well known algorithms. Studying what types of entanglement do or do not appear during the algorithm, and how these types evolve, can be helpful to understand more precisely the role and nature of entanglement in quantum computation.

~

In our study, we focus on the entanglement of four-qubit states in Grover's and Shor's algorithms. The four-qubit case is interesting because it contains an infinite number of SLOCC orbits, but there is still a classification  in terms of families of normal forms. Moreover, with the development of small quantum devices (IBM Quantum Experience), it is nowadays possible to implement famous quantum algorithms, such as Grover's and Shor's, with a limited number of qubits.

~

The paper is organized as follows. In Section \ref{four-qubits} we recall the four-qubit classification and related developed tools for classification purposes. In Section \ref{Resultgrover} we focus on Grover's algorithm  and present our results for the four-qubit case. In Section \ref{shor_algo} we study entanglement of periodic states appearing in Shor's algorithm, before and after the application of Quantum Fourier Transform. In the same section, we also brievely study the influence of QFT on the entanglement of general quantum states. Finally, we discuss our results and present a conclusion in Section \ref{conclusion}. The main observation raised by our study is that not all SLOCC classes appear and it is reasonable to wonder why some states, like the $\ket{W}$ state, never show up in those examples.

\section{Entanglement of four-qubits}\label{four-qubits}

\subsection{Verstraete \textit{et al.} classification}

The Hilbert space of 4-qubits $\calH = \CC^2 \otimes \CC^2 \otimes \CC^2 \otimes \CC^2$ contains an infinite number of orbits under the action of the SLOCC group $G = GL_2(\CC) \times  GL_2(\CC) \times  GL_2(\CC) \times  GL_2(\CC) $. 

~

In 2002, Verstraete \textit{et al.} \cite{Verstraete} used an original approach by studying the action of the group $SO(4)\times SO(4)$ as a subgroup of $SO(8)$ and generalizing the singular-value decomposition in matrix analysis to complex orthogonal equivalence classes. They proposed a list of 9 normal forms depending on parameters, which, up to permutation of the qubits, permit the parametrization of all SLOCC orbits. This list was corrected in 2006 by Chterental and Djokovic \cite{Chterental}, and this one was used for our work (see Table \ref{tableVerstraete}). 

\begin{table}[!h] 
 \begin{center}
\begin{tabular}{c}
\hline \hline \\
$ G_{abcd} = \frac{a+d}{2}\big(\ket{0000}+\ket{1111}\big) + \frac{a-d}{2}\big(\ket{0011}+\ket{1100}\big) + \frac{b+c}{2}\big(\ket{0101}+\ket{1010}\big) + \frac{b-c}{2}\big(\ket{0110}+\ket{1001}\big)$ \\
\\
$ L_{abc_2} = \frac{a+b}{2}\big(\ket{0000}+\ket{1111}\big) + \frac{a-b}{2}\big(\ket{0011}+\ket{1100}\big)  + c \big(\ket{0101}+\ket{1010}\big) + \ket{0110} $ \\
\\
$ L_{a_2b_2} = a \big(\ket{0000}+\ket{1111}\big) + b \big(\ket{0101}+\ket{1010}\big) + \ket{0011} + \ket{0110} $ \\
\\
$ L_{ab_3} =  a \big(\ket{0000}+\ket{1111}\big) + \frac{a+b}{2}\big(\ket{0101}+\ket{1010}\big) + \frac{a-b}{2}\big(\ket{0110}+\ket{1001}\big) + \frac{i}{\sqrt{2}}\big( \ket{0001} + \ket{0010} - \ket{0111} - \ket{1011} \big)    $ \\
\\
$ L_{a_4} = a \big(\ket{0000}+ \ket{0101} + \ket{1010} + \ket{1111}\big)+ i\ket{0001} + \ket{0110} - i\ket{1011} $ \\
\\
$ L_{a_20_{3\oplus \overline{1}}} = a \big(\ket{0000}+\ket{1111}\big) + \ket{0011} + \ket{0101} + \ket{0110} $ \\
\\
$ L_{0_{5\oplus \overline{3}}} = \ket{0000} + \ket{0101} + \ket{1000} + \ket{1110} $ \\
\\
$ L_{0_{7\oplus \overline{1}}} =  \ket{0000} + \ket{1011} + \ket{1101} + \ket{1110} $ \\
\\
$ L_{0_{3\oplus \overline{1}}0_{3\oplus \overline{1}}} = \ket{0000} + \ket{0111}$ \\
\\
\hline \hline
\end{tabular}
\end{center}
\caption{The nine (corrected) Verstaete \textit{et al.} forms.}
\label{tableVerstraete}
\end{table}

\subsection{Algebraic geometry and entanglement}

Let us denote by  $\ket{j_1} \otimes \ket{j_2} \otimes \ket{j_3} \otimes \ket{j_4}$ the standard basis, with $j_i \in [\![ 0,1 ]\!]$, of $\calH$. By shortening the basis notation to $\ket{j_1 j_2 j_3 j_4}$, we can write a 4-qubit state $\ket{\Psi}$ as: 

\begin{equation}
\ket{\Psi} = \sum_{j_1, j_2 ,j_3, j_4 \in  [\![ 0,1 ]\!]} a_{j_1 j_2 j_3 j_4} ~ \ket{j_1 j_2 j_3 j_4},
\end{equation}

with $a_{j_1 j_2 j_3 j_4} \in \CC$ and $\sum_{j_1, j_2 ,j_3, j_4 } | a_{j_1 j_2 j_3 j_4} | ^2 = 1$.

~

Nonzero scalar multiplication has no incidence on the entanglement nature of $\ket{\Psi} \in \calH$. Therefore, we can consider quantum states as points in the projective space $\PP(\calH) = \PP^{15}$. The set of separable states correspond to the set of factorized tensors, i.e. states that can be written $\ket{\Psi_{Sep}} = \ket{v_1} \otimes \ket{v_2} \otimes \ket{v_3} \otimes \ket{v_4} $, with $v_i = \alpha_i \ket{0} + \beta_i \ket{1} \in \CC^2$. 

~

If we look at the projectivization of the set of separable states, we retrieve an algebraic variety called the Segre embedding of the product of four projective lines $\PP(\CC^2)=\PP^1$, defined as the image of the Segre map: 

\begin{equation}
Seg :  \left\{\begin{array}{l} 
	~~~~~~ \PP^1 \times \PP^1 \times \PP^1 \times \PP^1 ~~~~~~~~ \to ~~~~~~~~~~~~~~~~\PP^{15}= \PP(\mathcal{H})\\
	~~~~~ ([v_1],[v_2],[v_3],[v_4]) ~~~~~~~~~ \mapsto ~~~~~~~~~~~ [v_1 \otimes v_2 \otimes v_3 \otimes v_4]
	  \end{array}\right. ,
\end{equation}

where $[v_i]$ refers to the projectivization of the vector $v_i$. When we work over $\PP(\calH)$, the group SLOCC corresponds to $G = SL_2(\CC) \times  SL_2(\CC) \times  SL_2(\CC) \times  SL_2(\CC)$.  The variety of separable states also corresponds to the orbit of highest weight vector, which can be chosen to be $\ket{0000}$, and will be denoted by $X$ such that 

\begin{equation}
X = Seg(\PP^1 \times \PP^1 \times \PP^1 \times \PP^1) = \PP(G \cdot \ket{0000}) ~ .
\end{equation}

Because $X$ is a $G$-orbit, auxiliary varieties, built from the knowledge of $X$, can be used to stratify the ambiant space by $G$-invariant varieties. This idea of using algebraic geometry to describe entanglement classes has been investigated several times in the past 15 years \cite{heydari, brylinski,HLT,HLT2,HLT3,sanz,sawicki1,sawicki2}. For instance in \cite{Miyake1,Miyake2,Miyake3}, the dual variety of $X$ was introduced to separate different classes of entanglement. Let us recall the definition of the dual of $X$,

\begin{equation}
X^*=\overline{\{H\in \PP(\mathcal{H}^*), \exists x\in X, T_xX\subset H\}},
\end{equation}

where $T_x X$ denotes the tangent space to $X$ at $x$. The dual variety parametrizes the hyperplanes tangent to $X$. In the $4$-qubit case the defining equation of the dual of $X$ is nothing but the so-called $2\times 2\times 2\times 2$ Cayley H
yperdeterminant \cite{LT}. In \cite{HLT3} the correspondence between the singular loci of the $2\times 2\times 2\times 2$ hyperdeterminant and 
the Verstraete and al's normal form was established (Figure \ref{onion}).
In the next section we explain the algorithm that we will use to identify the normal form of a given four qubit state. Geometrically this algorithm allows one to know to which strata of the singular locus of the hyperdetermiant a given state $\ket{\Psi}$ belongs. We also reproduced (Figure \ref{3onion}) the orbit stratification of the 3-qubit classification as we will give illustrative examples regarding this case in Section \ref{3qubit_shor_eett} and Appendix \ref{appendixC}. As pointed out by Miyake \cite{Miyake1}, the stratification of the orbits in the 3-qubit case can also be described in terms of singularities of the $2\times 2\times 2$ Cayley Hyperdeterminant.

\begin{figure}[!h]
 \includegraphics[width=8cm]{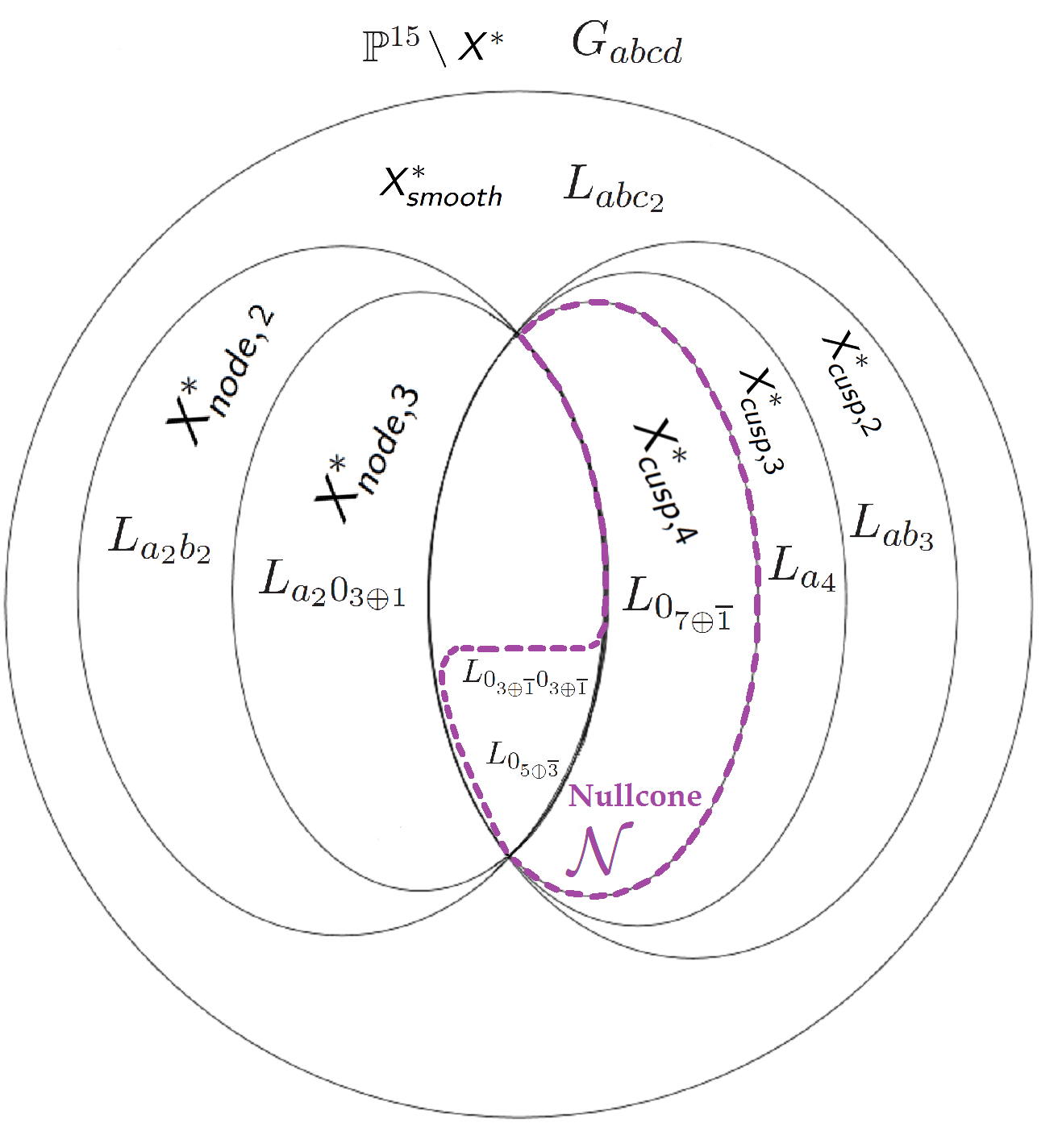}
\caption{Onion like structure of the entanglement in the 4-qubit case: The $G_{abcd}$ states do not belong to the dual variety $X^*$ for $a,b,c,d$ generic while states of type $L_{abc_2}$ corresponds to smooth point of $X^*$ for $a,b,c$ generic. The other states are points of the singular locus of $X^*$. The node components refer to states that corresponds to hyperplane with several points of tangency while the cusp components correspond to states/hyperplanes with one point of tangency of higher order \cite{HLT3}.}
\label{onion}
 \end{figure}
 
\begin{figure}[!h]
 \includegraphics[width=7cm]{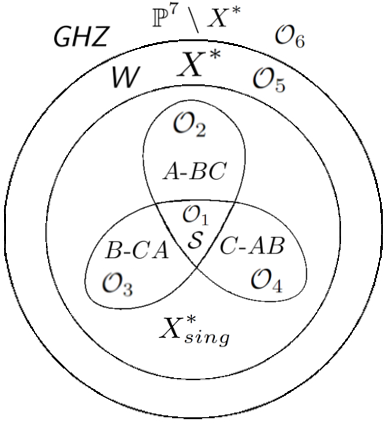}
\caption{Onion-like structure of the 3-qubit classication: The number of orbits is finite $\mathcal{O}_1,\dots,\mathcal{O}_6$ and the closures of each define algebraic varieties. For instance $\overline{\mathcal{O}_6}=\PP^7$ and $\overline{\mathcal{O}_5}=X^*$ (see \cite{HLT}).}
\label{3onion}
 \end{figure}
 \begin{rem}
  Other auxiliary varieties can be built from the knowledge of $X=Seg(\PP^1\times\PP^1\times\PP^1\times\PP^1)$ that are meaningful in terms of entanglement. Let us mention the secant variety of $X$, $\sigma(X)=\overline{\cup_{x,y\in X} \PP^1_{xy}}$ and the tangential variety $\tau(X)=\cup_{x\in X} T_x X$. When $X=Seg(\PP^1\times\PP^1\times\PP^1\times\PP^1)$, one has $\sigma(X)=\PP(\overline{\SLOCC.\ket{GHZ}_4})$ and $\tau(X)=\PP(\overline{\text{SLOCC}.\ket{W_4}})$ where the overline denotes the Zariski closure. See \cite{HLT,HLT2, HLT3} for more details about geometic constructions based on auxilary varieties to describe entanglement.
 \end{rem}

\subsection{Algorithms based on invariants}

We now introduce two algorithms based on a family of invariants and covariants that have been proposed in \cite{HLT2,HLT3}  to discriminate both the Verstaete \textit{et al.} families and the states that belong to the nullcone (states that vanish all SLOCC invariant polynomials).
These algorithms have also been used recently in \cite{enriquez,bataille}. Recall that one can associate to each 4-qubit state $\ket{\Psi}$ a quadrilinear form $A$:

\begin{equation}
\ket{\Psi} = \sum_{i,j,k,l \in  [\![ 0,1 ]\!]} a_{ijkl} ~ \ket{ijkl}, 
\end{equation}
\begin{equation}
A = \sum_{i,j,k,l \in  [\![ 0,1 ]\!]} a_{ijkl} \cdot x_iy_jz_kt_l ~ .
\end{equation}

The ring of  invariant polynomials for 4-qubits can be generated thanks to the four polynomial invariants $H$, $L$, $M$ and $D_{xy}$ \cite{LT}, defined as 

\begin{equation}
H = a_{0000}a_{1111} - a_{1110}a_{0001} - a_{0010}a_{1101} + a_{1100}a_{0011} - a_{0100}a_{1011} + a_{1010}a_{0101} + a_{0110}a_{1001} - a_{1000}a_{0111},
\end{equation}
\begin{equation}
L = \begin{vmatrix}
a_{0000} & a_{0010} & a_{0001} & a_{0011} \\ 
a_{1000} & a_{1010} & a_{1001} & a_{1011} \\ 
a_{0100} & a_{0110} & a_{0101} & a_{0111} \\ 
a_{1100} & a_{1110} & a_{1101} & a_{1111}  
\end{vmatrix},
\end{equation}
\begin{equation}
M =  \begin{vmatrix}
a_{0000} & a_{0001} & a_{0100} & a_{0101} \\ 
a_{1000} & a_{1001} & a_{1100} & a_{1101} \\ 
a_{0010} & a_{0011} & a_{0110} & a_{0111} \\ 
a_{1010} & a_{1011} & a_{1110} & a_{1111}  
\end{vmatrix},
\end{equation}
\begin{equation}
N = - L - M ,
\end{equation}
\begin{equation}
b_{xy} = \det \Big(\frac{\partial^2 A}{\partial z_i \partial t_j}\Big)_{i,j \in  [\![ 0,1 ]\!]} = [x_0^2,x_0x_1,x_1^2] ~ B_{xy} \begin{bmatrix} y_0^2 \\ y_0y_1 \\ y_1^2 \end{bmatrix} ~ .
\end{equation}
We determine the $3 \times 3$ matrix $B_{xy}$ by retrieving the coefficients in front of the terms $x_ix_jy_ky_l$ with $i,j,k,l \in  [\![ 0,1 ]\!]$ in the quadratic form $b_{xy}$.
\begin{equation}
D_{xy} = - \det (B_{xy}) ~ .
\end{equation}

~

One also needs to define three quartics , which coefficients are defined using the four qubit invariants:

\begin{equation}
\mathcal{Q}_1(\ket{\Psi}) = x^4 - 2H\cdot x^3y + (H^2 + 2L+4M) \cdot x^2y^2 + (4D_{xy} - 4H(M+\frac{1}{2} L) ) \cdot xy^3 + L^2 \cdot y^4,
\end{equation}
\begin{equation}
\mathcal{Q}_2(\ket{\Psi}) = x^4 - 2H\cdot x^3y + (H^2 - 4L -2M) \cdot x^2y^2 + (-2MH + 4D_{xy}) \cdot xy^3 + M^2 \cdot y^4,
\end{equation}
\begin{equation}
\mathcal{Q}_3(\ket{\Psi}) = x^4 - 2H\cdot x^3y + (H^2 + 2L -2M) \cdot x^2y^2 - (2LH +2MH  - 4D_{xy} ) \cdot xy^3 + N^2 \cdot y^4 ~ .
\end{equation}

~

We define also two invariant polynomials of the quartics, $I_2$ and $I_3$, as follows:

\begin{equation}
I_2(\mathcal{Q}_1) = I_2(\mathcal{Q}_2)= I_2(\mathcal{Q}_3) =\frac{4}{3} L^2 - \frac{4}{3} H^2 M  +  \frac{4}{3} LM  +  \frac{4}{3}M^2 + 2 H D_{xy} + \frac{1}{12} H^4 - \frac{2}{3}H^2 L,
\end{equation}
\begin{equation}
\begin{split} I_3(\mathcal{Q}_1) = I_3(\mathcal{Q}_2)= I_3(\mathcal{Q}_3)  = \frac{8}{27} L^3 - \frac{1}{216} H^6 - \frac{8}{27} M^3-  \frac{1}{6} Dxy H^3 + \frac{4}{3}H M Dxy - \frac{5}{9} H^2 M L + \frac{2}{3} H L Dxy \\ 
- \frac{2}{9} H^2 L^2 - \frac{5}{9} H^2 M^2 - {D_{xy}}^2 + \frac{4}{9} L^2 M + \frac{1}{18} H^4 L + \frac{1}{9} H^4 M - \frac{4}{9} L M^2
\end{split} ~ .
\end{equation}

~

The hyperdeterminant of a 4-qubit state $\ket{\Psi}$ can be seen as the discriminant of one of the quartic $Det_{2222}=\Delta(\mathcal{Q}_i)$, which is equal also to $Det_{2222}={I_2}^3 - 27{I_3}^2$. If we want to investigate the nature and multiplicity of the roots of the quartics, we have to use two other covariant polonomials, the Hessian and the Jacobian of the Hessain, defined as follows:

\begin{equation}
Hess(Q) =  \begin{vmatrix}
\frac{\partial^2 Q}{\partial x^2 } & \frac{\partial^2 Q}{\partial x \partial y } \\ 
\frac{\partial^2 Q}{\partial y \partial x } & \frac{\partial^2 Q}{\partial y ^2}
\end{vmatrix},
\end{equation}
\begin{equation}
T(Q) = \begin{vmatrix}
\frac{\partial Q}{\partial x } & \frac{\partial Q}{ \partial y } \\ 
\frac{\partial Hess(Q)}{\partial x } & \frac{\partial Hess(Q)}{ \partial y }
\end{vmatrix} ~ .
\end{equation}

~

The algorithms presented in \cite{HLT2,HLT3} consider two main cases. The first case is when the state $\ket{\Psi}$ belongs to the nullcone. The nullcone is defined as the set of states which annihilate all invariant polynomials. In practice, we can define the projectivization $\mathcal{N}$ of the nullcone as

\begin{equation}
\mathcal{N} = \{ \ket{\Psi} \in \PP(\calH) ~/ ~ H(\ket{\Psi}) = L(\ket{\Psi}) = M(\ket{\Psi}) = D_{xy}(\ket{\Psi}) = 0 \} ~ .
\end{equation}

The nullcone contains $31$ SLOCC-orbits. If one allows permutations of the four qubits by the symmetric group $\mathcal{S}_4$, those $31$ orbits can be grouped in $8$ non-equivalent strata of orbits $Gr_1, Gr_2,\dots, Gr_8$ forming a nested sequence (the orbit closures of the strata $Gr_{i+1}$ containing orbits of $Gr_i$). In particular $Gr_1$ only contains the orbit of separable states.

~

To distinguish between the different stratas of the nullcone, we will use these polynomials defined as the sum or product of covariants:

\begin{equation}
P_B = B_{2200} + B_{2020} + B_{2002} + B_{0220} + B_{0202} + B_{0022} ,
\end{equation}
\begin{equation}
P_C ^1= C_{3111} + C_{1311} + C_{1131} + C_{1113},
\end{equation}
\begin{equation}
P_C ^2= C_{3111} \cdot C_{1311} \cdot C_{1131} \cdot C_{1113},
\end{equation}
\begin{equation}
P_D ^1= D_{4000} + D_{0400} + D_{0040} + D_{0004},
\end{equation}
\begin{equation}
P_D ^2= D_{2200} + D_{2020} + D_{2002} + D_{0220} + D_{0202} + D_{0022}, 
\end{equation}
\begin{equation}
P_F = F^1_{2220} + F^1_{2202} + F^1_{2022} + F^1_{0222},
\end{equation}
\begin{equation}
P_L= L_{6000} + L_{0600} + L_{0060} + L_{0006} ~ .
\end{equation}

~

We can in fact decide to which strata a given from belongs by evaluating the vector $V$ defined in Equation (\ref{vect_cov}). When the evaluated value is non-zero, we replace the value by  '1'. The elements in the vector $V$ will thus only take binary values.
\begin{equation}\label{vect_cov}
V = [A,P_B,P_C^1,P_C^2,P_D^1,P_D^2,P_F,P_L] ~ .
\end{equation}

 ~

In Algorithm \ref{algo1}, we reproduce \cite{HLT2} a procedure that takes in input a 4-qubit state and the return the corresponding strata in the nullcone, or an error if the state is not nilpotent. 
 
\begin{algorithm}[H]
\caption{NilpotentType \cite{HLT2}}\label{algo1} 
\begin{algorithmic}
\Require $Y$ an array of size 16, the 4-qubit state
\Ensure The nullcone type $\mathcal{N}$ of $Y$

~

\If{isInNullcone($Y$)}
\State vectCov $\leftarrow [A,P_B,P_C^1,P_C^2,P_D^1,P_D^2,P_F,P_L]$
\State eval $\leftarrow$ evaluate(vectCov,$Y$)
\If{eval $= [0, 0, 0, 0, 0, 0, 0, 0]$}
\State \textbf{return} $Gr_0$
\ElsIf{eval $= [1, 0, 0, 0, 0, 0, 0, 0]$}
\State \textbf{return} $Gr_1$
\ElsIf{eval $= [1, 1, 0, 0, 0, 0, 0, 0]$}
\State \textbf{return} $Gr_2$
\ElsIf{eval $= [1, 1, 1, 0, 0, 0, 0, 0]$}
\State \textbf{return} $Gr_3$
\ElsIf{eval $= [1, 1, 1, 0, 1, 0, 0, 0]$}
\State \textbf{return} $Gr_4$
\ElsIf{eval $= [1, 1, 1, 1, 0, 0, 0, 0]$}
\State \textbf{return} $Gr_5$
\ElsIf{eval $= [1, 1, 1, 1, 1, 1, 0, 0]$}
\State \textbf{return} $Gr_6$
\ElsIf{eval $= [1, 1, 1, 1, 1, 1, 1, 0]$}
\State \textbf{return} $Gr_7$
\ElsIf{eval $= [1, 1, 1, 1, 1, 1, 1, 1]$}
\State \textbf{return} $Gr_8$
\EndIf
\Else
\State \textbf{return} ''Y does not belong to the nullcone''
\EndIf
\end{algorithmic}
\end{algorithm}

~

If the form is not nilpotent (does not belong to the nullcone), one needs to use other covariants computed in \cite{HLT2} to distinguish the different families. With the notations of \cite{HLT2} these are:

\begin{equation}
\mathcal{L} = L_{6000} + L_{0600} + L_{0060} + L_{0006},
\end{equation}
\begin{equation}
\mathcal{K}_5 = K_{5111} + K_{1511} + K_{1151} + K_{1115},
\end{equation}
\begin{equation}
\mathcal{K}_3 = K_{3311} + K_{3131} + K_{3113} + K_{1331} + K_{1313} + K_{1133},
\end{equation}
\begin{equation}
\mathcal{G} = G_{3111}^2 + G_{1311}^2 + G_{1131}^2 +G_{1113}^2,
\end{equation}
\begin{equation}
\overline{\mathcal{G}} = G_{3111}^1  G_{1311}^1  G_{1131}^1 G_{1113}^1,
\end{equation}
\begin{equation}
\mathcal{D} = D_{4000} + D_{0400} + D_{0040} + D_{0004},
\end{equation}
\begin{equation}
\mathcal{H} = H_{2220} + H_{2202} + H_{2022} + H_{0222},
\end{equation}
\begin{equation}
\mathcal{C} = G_{1111}^2 ~ .
\end{equation}

~

In order to determine the Verstraete \textit{et al.} type or family of a given state, we will use Algorithm \ref{algo2} reproduced in Appendix \ref{appendixB}. It is based on the roots of the three quartics and their multiplicities, and also based on the evaluation of the covariants on the form defined by the state.

\section{Entanglement in Grover's algorithm}\label{Resultgrover}

\subsection{Grover's algorithm}

In 1996, Lov Grover discovered a quantum algorithm for searching elements in a large and non-ordered database \cite{Grover} with a complexity of $\mathcal{O}(\sqrt{N})$ request, instead of the classical $\mathcal{O}(\frac{N}{2})$, with $N$ being the database size.

 ~
 
 There are four main steps in this algorithm, and they correspond to the following operations : Initialization, Oracle operator, Diffusion operator, Measurement. Depending on the number $|\mathcal{S}|$ of searched elements and the size $N$ of the database, the Oracle and Diffusion operators should be repeated several times. 
		
  \begin{figure}[!h]
\begin{center}
  \includegraphics[width = 12.0cm]{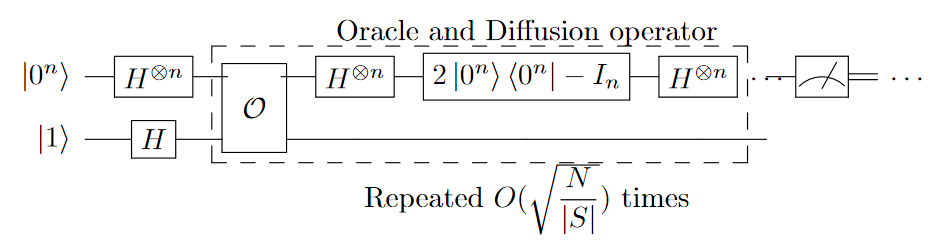}
  \caption{ Grover's algorithm as a quantum circuit.}\label{grovercircuit}
\end{center}
\end{figure}

~

The first register of $n$-qubits is initialized in the superposed state $\ket{+}^{\otimes n}$, then, depending on the searched elements, the Oracle gate (denoted by $\mathcal{O}$ in Figure \ref{grovercircuit}) will mark these elements (with a minus sign), and thus modify the  state of the first register. Then we will apply the Diffusion gate, in order to amplify the module of the amplitudes of the searched elements.

 ~

In our work, we focus on the different entanglement classes that can be generated by Grover's algorithm, in the four-qubit case. In other words, we want to know what are the different Verstraete \textit{et al.} families that can be reached, starting from the state $\ket{+}^{\otimes n}$ and applying Oracle/Diffusion gates by varying the number and indices of searched elements in Grover's algorithm. 

\subsection{Previous work}

Grover's algorithm is one of the most famous quantum algorithm in the litterature, because it provides better performances for searching elements in a database than what is done in the classical case. It has been proved that entanglement is involved during the steps of the algorithm, but his role and his nature has not yet been fully understood. In fact, Braunstein and Pati was the first to prove in 2000 the precense and the necessity of entanglement in Grover's algorithm \cite{Braunstein}.

~

In 2002, Biham, Nielsen and Osborne \cite{Biham} introduced an entanglement monotone derived from Grover's algorithm. Given as input a pure quantum state, the author study the maximization of $P_{max}$ the probability that the algorithm suceeds, and this, under local unitary operations. This defines an entanglement monotone, and in fact it defines the well known Groverian Measure of Entanglement. A recent 

~

The same year, Forcer \textit{et al.} \cite{Forcer} examined the roles played by superpostition and entanglement in quantum computing. The analysis is illustrated by a discussion on a classical implementation of Grover's algorithm. The absence of multi-particle entanglement leads to exponentially rising ressources for implementing such quantum algorithms. They conclude that multi-particle entanglement is the key property of quantum systems that provides the remarkable power of quantum computers. 

~

Biham, Shapira and Shimoni have analyzed, in 2003, the dynamics of Grover's algorithm while initializing the algorithm with an arbitrary $\ket{\Psi}$ pure state, instead of the $\ket{0}$ \cite{Biham2}. The authors showed that the optimal measurement time is the same with both initial states, in the case of the same number of marked elements. Biham \textit{et al.} generalized the Groverian entanglement measure to the case of several marked elements, previously limited to a single marked state. According to the authors, as long as $r << N$, with $r$ the number of marked elements and $N$ the size of the database, the Groverian measure is independent of $r$.

~

In 2004, Orus and Latorre investigated the scaling of entanglement in adiabatic version of Grover's algorithm \cite{Orus3}, precising that the Von Neumann entropy remains a bounded quantity regardless of the size of the system, even at they called the critical point. More precisely, "the maximum of entropy approches 1 as a square root in the size of the system, which is the typical Grover scaling factor".

~

In 2005, Fang \textit{et al.} \cite{Fang} studied the degree of entanglement present in a multi-qubit register during the algorithm process, under the formalism of density matrices. The authors analyzed the variations of the concurrence and Von Neumann entropy for the one and two-qubit reduced density matrix from an $n$-qubit register, in terms of the number of iterations, and this for one element marked. With the computed concurrence, they found that it can be related to the probability of success. Besides, they observed that the concurrence reach a maximum value at approximately half of the optimal number of iterations.

~

In 2008, Iwai \textit{et al.} \cite{Iwai} published an article dealing with the measure of entanglement with respect to bipartite partition of $n$-qubit, defined and studied from the viewpoint of Riemann geometry in a previous work of the first author \cite{Iwai2}. This previous work permits to establish the distance between the maximally entangled states and the separable ones. In this work, the authors determine the set of maximally entangled states nearest to a separable state as we can encounter at the beginning of Grover's algorithm. They confirmed the fact that while the initial and the marked elements are separable the algorithm generates a sequence of entangled states. 

~

In 2012, Wen and Cao \cite{Wen} explored the behaviors of multipartite entanglement in Grover's algorithm by looking at the adiabatic version of the quantum search algorithm. The authors calculated the Von Neumann entropy of all possible bipartite divisions of the system. This leaded to a new measure of multipartite entanglement, generalizing the measure introduced by Meyer and Wallach in \cite{Meyer}. They showed the existence and the evolution of multipartite entanglement during the process. Besides, similar scaling behaviors were observed both in bipartite and multipartite case (beginning  from zero to a maximum value at a critical point, then back to zero again). The symmetry behavior of the entanglement during the adiabatic evolution is also shown.  

~

Rossi \textit{et al.} studied in 2013 the scale invariance of entanglement dynamics in Grover's algorithm \cite{Rossi}. They calculated the amount of entanglement of the quantum states generated by Grover's algorithm, through the computational steps. They used for that a numerical measure of entanglement named GME (Geometric Measure of Entanglement). They showed that multi-partite entanglement is always present during the algorithm, and studied it in function of the number of iterations of the algorithm, for a fixed number of qubits. They found the maximum of entanglement to be at some remarkable position, using the GME, and this, for specific values of marked elements.

~

The same year, Chakraborty \textit{et al.} published a work \cite{Chakraborty} in the same philosophy as the work of Rossi \textit{et al.} \cite{Rossi}. The Geometric Measure of Entanglement has been used to quantify and analyse entanglement across iterations of the algorithm. The authors investigated how entanglement varies when the number of qubits and marked elements increases. The first result is that the behavior of the maximum value of entanglement is monotonous with the number of qubits. The second main result is that, for a given number of qubits, a change in marked elements alters the amount of entanglement. 

~

One year later, Rossi \textit{et al.} came back with another work involving Grover states and hypergraphs \cite{Rossi2}. Numercial computation of the GME as a function of the number of qubit was made, when they only considered the first iteration of the algorithm. For different cases under consideration, the curves for one and two marked elements show the same behavior, i.e. an exponential decay. A link between the initial states of the algorithm and hypergraphs is established, giving a more pictural representation, and permetting to highlight some entanglement properties of these states, such as biseparability and the presence of genuine multipartite entanglement. 

~

In 2015, Qu \textit{et al.} \cite{Qu} investigated multipartite entanglement by using separable degree as a qualitative measure. On another hand, they also used a quantitative measure of entanglement introduced by Vidal \cite{Vidal}, namely the Schmidt number. These qualitative and quantitive tools permited to study the entanglement dynamics of Grover's search algorithm. The results, depending on the step in the algorithm, confirm that after first iterations fully entangled states appear in the algorithm. 

~

Recently, in 2016, Ye \textit{et al.} \cite{Ye} published a work dealing with the influence of static imperfections on quantum entanglement and quantum discord. They used concurrence to investigate the behavior of entanglement. Static imperfections can break quantum correlations, according to the authors. in fact, for every weak imperfections, the quantum entanglement exhibit periodic behavior, while the periodicity will be destroyed with stronger imperfections. They confirmed therefore the periodic property of entanglement involved in Grover's algorithm. 

~

The same year, in 2016, the authors of the present work, with Ismael Nounouh, investigated the entanglement nature of quantum states generated of Grover's algorithm by means of algebraic geometry \cite{HJN}. The authors established a link between entanglement of states generated by the algorithm and auxiliary algebraic varieties built from the set of separable states. We were able to propose a qualitative interpretation of the earlier results, such as the work of Rossi \textit{et al.} \cite{Rossi,Rossi2}. Some examples were investigated, such as the 3-qubits case. 

~

In 2017, Pan \textit{et al.} also studied the GME scale invariant property \cite{Pan}. Starting from the work of Rossi \textit{et al.} \cite{Rossi}, the authors showed that the entanglement dynamics in Grover's algorithm is not always scale invariant. They showed that after the turning point, the GME is not necessarily scale invariant, and then depend on the number $n$ of qubits and the marked elements. Some examples, when the searched elements form separable states, GHZ or W states, were investigated to confirm the proposed results. 

~

In this work, we continue in the same direction as our previous work \cite{HJN}, trying to bring explanation to behaviors observed by the differents authors in the domain. More precisely, we study the example of 4-qubits case, detailing what types of entanglement can be generated by Grover's algorithm in function of the marked elements, and recovering some results established in our previous paper.

\subsection{The 4-qubit case}

We investigate entanglement nature of states involved in Grover's algorithm by varying the marked elements, and by determining the corresponding Verstraete \textit{et al.} family or the nullcone strata. 
We list bellow all Verstraete \textit{et al.} normal forms and the strata of the nullcone reached. $|\mathcal{S}|$ denotes the number of marked elements.

\begin{itemize}
	\item \textbf{Standard regime ($|\mathcal{S}|<\frac{N}{4}$):}
	\begin{itemize}
		\item For $|\mathcal{S}|=1$, we always reach the $G_{00cc}$ orbit, as expected.
		\item For $|\mathcal{S}|=2$, the states generated by Grover's algorithm belong to $G_{abc0}$, $L_{00c_2}$, $L_{ab0_2}$, $Gr_8$ and $Gr_4$.
		\item For $|\mathcal{S}|=3$, we can obtain the orbits $G_{abc0}$, $L_{00c_2}$, $L_{aa0_2}$, $L_{0_2b_2}$ and $L_{a_20_{3\oplus \overline{1}}}$.
	\end{itemize}
	\item \textbf{Critical case ($|\mathcal{S}|=\frac{N}{4}$):}  For $|\mathcal{S}|=4$, which is the critical case (all amplitudes, except the marked one, are sent to zero, and the algorithm converges after one iteration), we can reach all the states that can be written as a sum of 4 basis states: $G_{00cc}$, $G_{a000}$, $G_{ab00}$, $L_{00c_2}$, $L_{aa0_2}$, $L_{a00_2}$, $L_{0_2b_2}$, and from $Gr_8$ to $Gr_1$.
	\item \textbf{Exceptional case (}$|\mathcal{S}|>\frac{N}{4}$, this is a not the standard case of application of Grover because the number of marked elements is not small compared to $N${\bf):}
	\begin{itemize}
		\item For $|\mathcal{S}|=5$, the orbits $G_{abc0}$, $G_{ab00}$, $L_{00c_2}$, $L_{aa0_2}$, $L_{ab0_2}$, $L_{a_2b_2}$, $L_{0_2b_2}$ and $L_{a_20_{3\oplus \overline{1}}}$ can be obtained. 
		\item For $|\mathcal{S}|=6$, Grover's algorithm can generate states that belong to $G_{abcd}$, $G_{abc0}$, $G_{ab00}$,  $L_{00c_2}$, $L_{aa0_2}$, $L_{ab0_2}$, $L_{a_2b_2}$, $L_{0_2b_2}$, $L_{a_4}$, $L_{a_20_{3\oplus \overline{1}}}$, $Gr_8$ and $Gr_4$.
		\item For $|\mathcal{S}|=7$, one can reach the following orbits: $G_{abcd}$, $G_{abc0}$, $G_{00cc}$, $G_{ab00}$, $L_{00c_2}$, $L_{aa0_2}$, $L_{ab0_2}$, $L_{a_2b_2}$, $L_{0_2b_2}$, $L_{a_20_{3\oplus \overline{1}}}$.
		\item For $|\mathcal{S}|=8$, the generated states belong to $G_{abc0}$, $G_{00cc}$, $G_{ab00}$, $L_{00c_2}$, $L_{aa0_2}$, $L_{a_2b_2}$, $L_{0_2b_2}$, $L_{a_20_{3\oplus \overline{1}}}$, $Gr_8$, $Gr_7$, $Gr_4$, $Gr_2$ and $Gr_1$.
	\end{itemize}
\end{itemize}

 ~
 
In the case of one marked element, we always belong to the secant variety $\sigma_2(X)$, which is the equivalence class of the $\ket{GHZ_4}$ state, also know as the set of tensors of rank 2, as it was demonstrated by the authors in \cite{HJN}. 

~

We want to point out that, as it was mentionned in \cite{HJN}, the state $\ket{W_4}$ (which belongs to the $Gr_5$ orbit) is not reached by the states generated by Grover's algorithm, except in the critical case. The other orbits that are not reached, except in the critical case, are $Gr_6$ and $Gr_3$. Moreover, we remark that the generic families $L_{abc_2}$ and $L_{ab_3}$ are not reached by states generated by Grover's algorithm, for 4-qubits. We present in Appendix \ref{appendixA} examples of marked elements  generating a given family or orbit.

~

Besides, if we plot the variation, as a function of $k$ (number of iterations), of the absolute value of the 4-qubits hyperdeterminant evaluated on the state $\ket{\Psi_k}$ (the state generated by Grover's algorithm at the $k$-th iteration), one obtains two different curves illustrating the periodicity of the algorithm (Figs. \ref{hyperdet1},\ref{hyperdet2} ). Since this $2\times 2 \times 2 \times 2$ hyperderminant vanishes for states in the dual variety, plotting is only relevant for (sub-)families like $G_{abcd}$ and $G_{abc0}$.

  \begin{figure}[!h]
\begin{center}
  \includegraphics[width = 9cm]{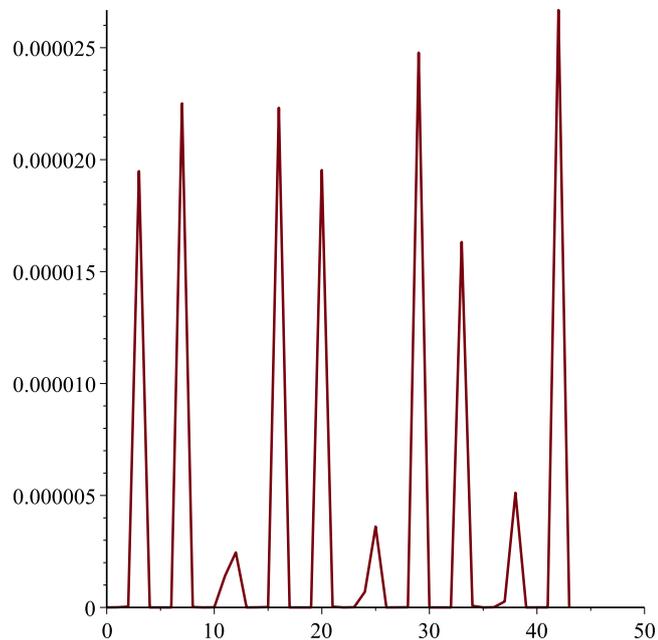}
  \caption{Evolution of the absolute value of the hyperdeterminant of 4-qubits in function of the number of iteration of Grover's algorithm, for the set of marked elements $S = \{ \ket{0000},\ket{1111}\}$ (states in $G_{abc0}$).}\label{hyperdet1}
\end{center}
\end{figure}

~

  \begin{figure}[!h]
\begin{center}
  \includegraphics[width = 9cm]{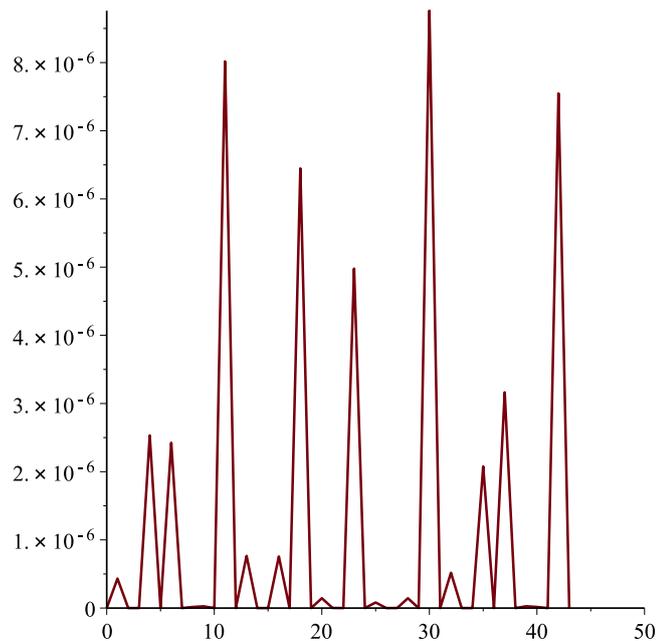}
  \caption{Evolution of the absolute value of the hyperdeterminant of 4-qubits in function of the number of iteration of Grover's algorithm, for the set of marked elements $S = \{ \ket{0000}, \ket{0001}, \ket{0010} ,\ket{0101}, \ket{1010}, \ket{1111} \}$ (states in $G_{abcd}$).}\label{hyperdet2}
\end{center}
\end{figure}

\section{Entanglement in Shor's algorithm} \label{shor_algo}

\subsection{Shor's algorithm}

In 1995, Peter Shor revealed a new quantum algorithm for integer factorization \cite{Shor}. The particularity of Shor's algorithm is that on a quantum computer, to factor an integer $M$, it runs in polynomial time. Thus, it can be used to break public-key cryptography schemes such as the widely used RSA scheme.

~

The algorithm is composed of several steps (pre-processing, order-finding, post-processing), but the only quantum part of the algorithm concerns the order-finding algorithm, so we will only focus on this part. 

~

We define a function $f : \mathcal{H}_N \to  \mathcal{H}_N$ with $\mathcal{H}_N = \{\ket{x} ~/~ x\in\mathbb{N},x<N\}$. We say $f$ is periodic of period $r < N$, when :
   \begin{equation}
   \forall x\in [\![ 0,N-r-1 ]\!] \textrm{, } f(\ket{x+r}) = f(\ket{x}) ~ .
   \end{equation}

For the period finding problem, one defines the periodic function $f$ as, $f(\ket{x}) = \ket{a^x \mod M}$, which takes the state $\ket{x}$ in parameter and return the state $\ket{a^x\text{ modulo }M}$. One also defines the related quantum gate $U_f : (x,y) \rightarrow (x,y\oplus f(x))$, with $x$ called the data register.

~
 
The order-fiding algorithm, which can be represented by a quantum circuit (see Figure \ref{order_fig}), is used to determine the order $r$ of $a$ modulo $M$, i.e. the smallest integer $r\in \mathbb{N}^*$ such as $a^r \equiv 1 ~[M]$.

~

We start by initializing the first register to the parallelized state $\ket{+}^{\otimes n}$, and the second register to the qubit $\ket{0}$. Then, we apply the $U_f$ gate and we make a measure on the second register. By knowing that the function $f$ is periodic, we will retrieve a periodic state for the first register. Finally, we apply the Quantum Fourier Transform to the first register, in order to extract some information about the period $r$, and we make a measurement on the first register. 

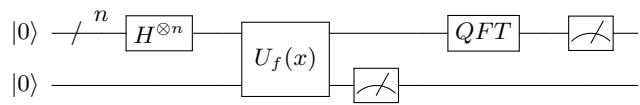
\begin{figure}[!h]
   \begin{align*}
   \Qcircuit @C=1em @R=.7em {
   \lstick{\ket{0}}  & {/}\qw & \ustick{n}\qw & \gate{H^{\otimes n}}   & \qw  & \multigate{1}{U_f(x)} & \qw & \qw  &\gate{QFT} &\qw &\meter &\qw\\
   \lstick{\ket{0}}  & \qw    & \qw           & \qw    &\qw   & \ghost{U_f(x)} & \meter & \qw  & \qw& \qw &\qw &\qw\\
   }
   \end{align*}
   \caption{Quantum circuit of order-fiding algorithm.}
   \label{order_fig}
   \end{figure}

~

In our work, we focused on the nature of entanglement for the state obtained after applying the $U_f$ gate and measuring the second register (which give us a periodic state), and after applying the Quantum Fourier Transform to the previous periodic state.

\subsection{Previous work}

Shor's algorithm is a quantum algorithm that offer an exponential speed-up over its classical counterparts. 

~

At the beginning of the 2000's, Jozsa \textit{et al.} suggested that quantum entanglement is playing a major role in quatum computational efficiency \cite{Jozsa,Jozsa2}. It has been proven in fact that quantum entanglement is involved in Shor's algorithm, theoretically \cite{Parker,Jozsa2,Shimoni,Kendon,Most} and experimentally \cite{Lanyon,Lu}.

~

Most of the studies related to entanglement, in this algorithm, focused on the entanglement between the two quantum registers.  In 2001, Parker and Planio \cite{Parker} looked at the average bipartite entaglement by using the logarithmic negativity, as a measure of the entanglement, at each stage of the algorithm. The states of the algorithm are defined by the controlled-$U_{\alpha}$ operations (composing the modulo exponentiation). They proved that entanglement exists in the algorithm and that the amount of entanglement increases towards the end of the algorithm. They also showed that if one tries to reduce the entanglement by introducing more mixing into the control qubit, one reduces  the efficiency of the computation. 

~

Two years later, Jozsa and Linden published an article dealing with the role of entanglement in quantum computational speed-up \cite{Jozsa2}. They discussed the difference between classical and quantum computation, and when a quantum computation can be efficiently classically simulated. In particular, it has been proposed that if we cannot efficiently classically simulate a quantum algorithm in polynomial time, then quantum entanglement is involved in this quantum algorithm. It is the case of Shor's algorithm, and the presence of entanglement is proved by considering "arithmetic progression" states (equivalent to periodic states) and the fact that they are not $p$-blocked. 

~

In 2004, Orus and Latorre \cite{Orus3} studied the scaling of entanglement in Shor's algorithm, proving analytically that it makes use of an exponentially large amount of entanglement between the two registers after the modular exponentiation step. This implies the impossibility of an efficient classical simulation using the protocol proposed by Vidal in \cite{Vidal}.

~

In 2005, Shimoni \textit{et al.} used the Groverian measure of entanglement to characterize the quantum states generated by Shor's algorithm \cite{Shimoni}. At each step of the QFT process (after each controlled-phase gates, Hadamard gate doesn't affect entanglement) they evaluate the Groverian measure of entanglement, and this for general quantum states and for periodic states of Shor's algorithm. For random product states, they showed that entanglement remains the same during most of the steps, but for particular controlled-phase gate the amount of entanglement measured by the Groverian measure increases significatively. For periodic states, it was found that the Groverian measure of entanglement doesn't change essentially, and the changes that we can encounter for small number $n$ of qubits vanish exponentially with $n$.

~

Kendon and Munro published in 2006 an article "Entanglement and its role in Shor's algorithm" \cite{Kendon} where they investigate the entanglement involved in Shor's algorithm by decomposing the $U_f$ gate and considering the $QFT^{-1}$ as a single gate. They focused first on the entanglement between the first and second register 	and then studied the entanglement involved in the first register. A quantitative study of entanglement is done in this paper, and some measures of entanglement like entropy of subsystems (between the two registers), negativity and entanglement of formation (within the first register) are used. According to the authors, after the modular exponentiation, the entanglement between the two registers cannot decrease during the $QFT^{-1}$. Furthermore, "entanglement within the upper register can only be generated or shifted around, not decreased". The authors also pointed out that the closer the period $r$ is to a power of 2, the smaller the value of the difference in the average entropy $\Delta E_1$ before and after the $QFT^{-1}$. When $r$ is a power of 2, then the $QFT^{-1}$ is exact giving $\Delta E_1=0$ in all cases.

~

In 2007, some experimentations of Shor's algorithm were implemented. Lanyon \textit{et al.} \cite{Lanyon} implemented a compiled version of Shor's algorithm by using a photonic system. They proved the existence of entanglement within the algorithm via quantum state and process tomography, and that entanglement is involved in the arithmetic calculations. The same year, Lu \textit{et al.} \cite{Lu} implemented  Shor's factoring algorithm also using photonic qubits. The experimentation was made with 4-photonic-qubits, and they detected genuine multiparticle entanglement during the algorithm (between the first and second register). 

~

Three years later, Most, Shimoni and Biham published a work \cite{Most} related with entanglement of Periodic States, the Quantum Fourier Transform and Shor's algorithm. They pointed out the importance and role of the periodic states during the algorithm. They also analysed the entanglement of periodic states, involved in Shor's algorithm, and they looked at how these states are affected by the Quantum Fourier Transform. Some approximations were used in order to evaluate the Groverian measure of entangled periodic states. According to the authors, the QFT does not change the entanglement of periodic states, for a sufficiently large number of qubits. 

~

All these studies have investigated the entanglement by using entanglement measure, and more precisely a quantitative analysis of the entanglement. It permited to give a first idea about how entanglement can evoluate during the algorithm. Here we focus on the periodic states and states after QFT.

\subsection{Entanglement of Periodic states}

In Shor's algorithm, the periodic states $\ket{\Psi^n_{l,r}}$ of $n$-qubits, with shift $l$ and period $r$, that we can encounter after measuring the upper register, can be written as : 

\begin{equation}\label{periodic} \ket{\Psi^n_{l,r}} = \frac{1}{\sqrt{A}} \sum_{i=0}^{A-1}\ket{l+ir},  ~~\text{with } A = \ceil*{\frac{N-l}{r}}, ~ N=2^n ~ .\end{equation}

For example, for the periodic 3-qubit states, with shift $l=3$ and period $r=2$, we have $A = \ceil*{\frac{8-3}{2}} = 3$ basis elements, so :

\begin{equation} \ket{\Psi^3_{3,2}} = \frac{1}{\sqrt{3}} \big(\ket{3}+\ket{5} + \ket{7}\big) = \frac{1}{\sqrt{3}}  \big(\ket{011}+\ket{101} + \ket{111}\big)  ~ .\end{equation}

~

For every $(l,r) \in [\![0,15 ]\!]\times[\![1,15]\!]$, we write the periodic state and compute the corresponding Verstraete family. All the results are given in the Table \ref{periodic4res}.

\begin{table}[!h]
 \begin{center}
  \begin{tabular}{|l||c|c|c|c|c|c|c|c|c|c|c|c|c|c|c|c|}
  \hline
   \backslashbox{$r$}{$l$} & 0 & 1 & 2 & 3 & 4 & 5 & 6 & 7 &  8 & 9 & 10 & 11 & 12 & 13 & 14 & 15   \\
   \hline
   \hline
  1 & $Gr_1$ & $G_{00cc}$ & $Gr_4$ & $L_{00c_2}$ & $Gr_2$ & $L_{00c_2}$ & $Gr_4$ & $G_{00cc}$ &  $Gr_1$ & $Gr_4$ & $Gr_2$ & $Gr_4$ & $Gr_1$ & $Gr_2$ & $Gr_1$ & $Gr_1$ \\
   \hline
  2 & $Gr_1$ & $Gr_1$ & $Gr_4$& $Gr_4$ & $Gr_2$ & $Gr_2$ & $Gr_4$ & $Gr_4$ &  $Gr_1$ & $Gr_1$ & $Gr_2$ & $Gr_2$ & $Gr_1$ & $Gr_1$ & $Gr_1$ & $Gr_1$ \\
   \hline
  3 & $G_{aa(-2a)0}$ & $G_{abc0}$ & $G_{abc0}$ & $L_{a00_2}$ & $Gr_8$ & $Gr_8$ & $L_{aa0_2}$ & $Gr_6$ &  $Gr_3$ & $Gr_3$ & $Gr_4$ & $Gr_2$ & $Gr_2$ & $Gr_1$ & $Gr_1$ & $Gr_1$ \\
   \hline
  4 & $Gr_1$ & $Gr_1$ & $Gr_1$ & $Gr_1$ & $Gr_2$ & $Gr_2$ & $Gr_2$ & $Gr_2$ &  $Gr_1$ & $Gr_1$ & $Gr_1$ & $Gr_1$ & $Gr_1$ & $Gr_1$ & $Gr_1$ & $Gr_1$ \\
   \hline
  5 & $G_{a000}$ & $Gr_6$ &  $Gr_6$ &  $Gr_6$ &  $Gr_6$ & $L_{aa0_2}$ & $Gr_4$ & $Gr_4$ &  $Gr_2$ & $Gr_4$ & $Gr_2$ & $Gr_1$ & $Gr_1$ & $Gr_1$ & $Gr_1$ & $Gr_1$  \\  
  \hline
  6 & $Gr_3$ & $Gr_3$ & $Gr_3$ & $Gr_3$ & $Gr_4$ & $Gr_4$ & $Gr_2$ & $Gr_2$ & $Gr_2$ & $Gr_2$ & $Gr_1$ & $Gr_1$ & $Gr_1$ & $Gr_1$ & $Gr_1$ & $Gr_1$  \\  
  \hline
  7 & $Gr_6$ & $Gr_6$ & $Gr_4$ & $Gr_2$ & $G_{00cc}$ & $Gr_2$ & $Gr_4$ & $Gr_2$ &  $Gr_4$ & $Gr_1$ & $Gr_1$ & $Gr_1$ & $Gr_1$ & $Gr_1$ & $Gr_1$ & $Gr_1$  \\  
  \hline
  8 & $Gr_1$ & $Gr_1$ & $Gr_1$ & $Gr_1$ & $Gr_1$ & $Gr_1$ & $Gr_1$ & $Gr_1$ &  $Gr_1$ & $Gr_1$ & $Gr_1$ & $Gr_1$ & $Gr_1$ & $Gr_1$ & $Gr_1$ & $Gr_1$  \\  
  \hline
  9 & $Gr_2$ & $Gr_4$ & $Gr_2$ & $G_{00cc}$ & $Gr_2$ & $Gr_4$ & $Gr_2$ & $Gr_1$ &  $Gr_1$ & $Gr_1$ & $Gr_1$ & $Gr_1$ & $Gr_1$ & $Gr_1$ & $Gr_1$ & $Gr_1$  \\  
  \hline
  10 & $Gr_2$ & $Gr_2$ & $Gr_4$ & $Gr_4$ & $Gr_2$ & $Gr_2$ & $Gr_1$ & $Gr_1$ &  $Gr_1$ & $Gr_1$ & $Gr_1$ & $Gr_1$ & $Gr_1$ & $Gr_1$ & $Gr_1$ & $Gr_1$ \\  
  \hline
  11 & $Gr_4$ & $Gr_4$ & $G_{00cc}$ & $Gr_4$ & $Gr_4$ & $Gr_1$ & $Gr_1$ & $Gr_1$ &  $Gr_1$ & $Gr_1$ & $Gr_1$ & $Gr_1$ & $Gr_1$ & $Gr_1$ & $Gr_1$ & $Gr_1$  \\  
  \hline
  12 & $Gr_2$ & $Gr_2$ & $Gr_2$ & $Gr_2$ & $Gr_1$ & $Gr_1$ & $Gr_1$ & $Gr_1$ &  $Gr_1$ & $Gr_1$ & $Gr_1$ & $Gr_1$ & $Gr_1$ & $Gr_1$ & $Gr_1$ & $Gr_1$  \\  
  \hline
  13 & $Gr_4$ & $G_{00cc}$ & $Gr_4$ & $Gr_1$ & $Gr_1$ & $Gr_1$ & $Gr_1$ & $Gr_1$ & $Gr_1$ & $Gr_1$ & $Gr_1$ & $Gr_1$ & $Gr_1$ & $Gr_1$ & $Gr_1$ & $Gr_1$  \\  
  \hline
  14 & $Gr_4$ & $Gr_4$ & $Gr_1$ & $Gr_1$ & $Gr_1$ & $Gr_1$ & $Gr_1$ & $Gr_1$ & $Gr_1$ & $Gr_1$ & $Gr_1$ & $Gr_1$ & $Gr_1$ & $Gr_1$ & $Gr_1$ & $Gr_1$  \\  
  \hline
  15 & $G_{00cc}$ & $Gr_1$ & $Gr_1$ & $Gr_1$ & $Gr_1$ & $Gr_1$ & $Gr_1$ & $Gr_1$ & $Gr_1$ & $Gr_1$ & $Gr_1$ & $Gr_1$ & $Gr_1$ & $Gr_1$ & $Gr_1$ & $Gr_1$ \\  
  \hline
  \end{tabular}
\caption{Verstaete \textit{et al.} families of periodic states depending on their shift $l$ and period $r$.}\label{periodic4res}
 \end{center}

\end{table}

From the results presented in this table, we can extract the following properties for 4-qubit periodic states:

\begin{itemize}
\item When $l=0$ and the period take the values $r=1$, $r=2$, $r=4$ and $r=8$, the state is a separable state,

\item When $r=8$, the state always belong to the $Gr_1$ orbit (separable),

\item All the states on the anti-diagonal starting from $\{l=1,r=15\}$ to $\{l=15,r=1\}$, and all at the bottom-right (under) this anti-diagonal belong to the $Gr_1$ orbit (separable),

\item If the period take the values $r=1$, $r=2$, $r=4$ and $r=8$, and if the shift $0\leq l \leq r-1$ or $\frac{N}{2} \leq l \leq \frac{N}{2}+ r-1$, then the state is a separable state,

\item When $\{l=1,r=1\}$, $\{l=7,r=1\}$,  and when $\{l=0,r=15\}$, $\{l=1,r=13\}$, $\{l=2,r=11\}$, $\{l=3,r=9\}$ and $\{l=4,r=7\}$ (almost half of an anti-diagonal) the periodic states belong to the $G_{00cc}$ orbit, which is the orbit of $\ket{GHZ_4}$.

\end{itemize}

These observations lead us to try a generalisation of these properties for any $n$-qubit periodic state.

\begin{proposition}
Let $\ket{\Psi^n_{l,r}}$ be a $n$-qubit periodic state. If the shift $l=0$ and the period $r=2^s$ divide $N=2^n$ then $\ket{\Psi^n_{0,r}}$ is a separable state, and it can be written $\ket{\Psi^n_{0,r}}= \ket{+}^{\otimes (n-s)} \otimes  \ket{0}^{\otimes s}$.
\end{proposition}

\begin{proof}
Let $\ket{\Psi^n_{l,r}}$ be a $n$-qubit periodic state, and let $N=2^n$. We suppose that the shift $l=0$, and that the period divide $N$. So there exist $p$ such that $r\times p = N$, and thus there exist $(s,q)\in \mathbb{N}^2$ such that $r=2^s$ and $p=2^q$. We recall that $\ket{\Psi^n_{0,r}}$ can be expressed as : 

\begin{equation}
\ket{\Psi^n_{0,r}} = \frac{1}{\sqrt{A}} \sum_{i=0}^{A-1}\ket{0+ir} ,  ~~\text{with } A = \ceil*{\frac{N-0}{r}} = p,
\end{equation}

so 

\begin{equation}
\ket{\Psi^n_{0,r}} = \frac{1}{\sqrt{p}} \big( \ket{0} + \ket{r} + \ket{2r} + \cdots + \ket{(p-1)r} \big) ~ .
\end{equation}

When the period $r$ is equal to $N$, we have $A = \ceil*{\frac{N-0}{r}} = 1$ basis state in the writing of $\ket{\Psi^n_{0,r}}$ , which is $\ket{0\dots 0}$ in the binary notation. So the state is separable and  we have $\ket{\Psi^n_{0,r}} = \ket{+}^{\otimes (n-n)} \otimes  \ket{0}^{\otimes n} = \ket{0}^{\otimes n}$. Besides, if the period $r$ is equal to one, then we have all the $A = \ceil*{\frac{N-0}{r}} = N$ basis states in the decomposition of the periodic state, so we get in fact the fully superposed state as expected : $\ket{\Psi^n_{0,r}} = \ket{+}^{\otimes (n-0)} \otimes  \ket{0}^{\otimes 0} = \ket{+}^{\otimes n}$.

~

Now we consider the case when $1<r<N$, i.e. when $2\leq r \leq \frac{N}{2}$. The periodic state is a sum of $p$ basis state, and $p$ is even, so we can always split the periodic state into two parts :

\begin{equation}
\ket{\Psi^n_{0,r}} = \frac{1}{\sqrt{p}} \big(  \underbrace{\ket{0} + \ket{r} + \ket{2r} + \cdots + \ket{\big(\frac{p}{2}-1\big)r}}_{\dfrac{p}{2} \text{ terms}}    + \underbrace{   \ket{\frac{p}{2}r}+ \cdots + \ket{(p-1)r}}_{\dfrac{p}{2} \text{ terms}} \big) ~ .
\end{equation}

We can easly see that we always have the $\ket{0}$ and $\ket{\dfrac{p}{2}r}=\ket{\dfrac{N}{2}}$ states in the parts of the periodic state, and these basis states can be written in binary notation $\ket{00\dots 0}$ and $\ket{10\dots 0}$. Then, depending on the period $r$ we will have other terms or not in each part, but we can always rewrite the periodic state as follow :

\begin{equation}
\ket{\Psi^n_{0,r}} = \frac{1}{\sqrt{p}} \big(  \underbrace{\ket{0} + \ket{r} + \ket{2r} + \cdots + \ket{\big(\frac{p}{2}-1\big)r}}_{\dfrac{p}{2} \text{ terms}}    + \underbrace{   \ket{\frac{N}{2}+0}+ \ket{\frac{N}{2}+r}+ \ket{\frac{N}{2}+2r}+\cdots + \ket{\frac{N}{2}+\big(\frac{p}{2}-1\big)r}}_{\dfrac{p}{2} \text{ terms}} \big) ~ .
\end{equation}

Therefore, the periodic state can be express as : 

\begin{equation}\label{periodictensor}
\ket{\Psi^n_{0,r}} = \ket{+} \otimes \frac{1}{\sqrt{2^{q-1}}} \Big( \ket{0} + \ket{r} + \ket{2r} + \cdots + \ket{\big(\frac{p}{2}-1\big)r} \Big) =  \ket{+} \otimes \ket{\Psi^{n-1}_{0,r}} ~ .
\end{equation}

Now if we consider the state $\ket{\Psi^{n-1}_{0,r}}$, we can repeat the same process until it remains only one state in the sum decomposition of the periodic state. This happen when $A=1$, so when the number of qubit of the periodic state considered (at the right side of the tensor product in the equation (\ref{periodictensor})) is equal to $s$. So this process is repeated $q=n-s$ times, and at the end, we retreive a seperable state: 

\begin{equation}
\ket{\Psi^n_{0,r}}= \ket{+}^{\otimes (n-s)} \otimes  \ket{0}^{\otimes s} ~ .
\end{equation}

\end{proof}

\begin{proposition}
Let $\ket{\Psi^n_{l,r}}$ be a $n$-qubit periodic state. If the period $r$ is equal to $\frac{N}{2}=2^{n-1}$ then, for all possible values of the shift $l$, $\ket{\Psi^n_{l,r}}$ is a separable state.

\end{proposition}

\begin{proof}
Let $\ket{\Psi^n_{l,r}}$ be a $n$-qubit periodic state, and let $N=2^n$. We suppose that the period $r$ is equal to $\frac{N}{2}$. The periodic state will contain only $A = \ceil*{\dfrac{N-l}{\frac{N}{2}}}$ basis states. If the shift satisfies $0 \leq l <\frac{N}{2}$ then we will have $A=2$ elements, otherwise, if  $\frac{N}{2} \leq l < N$ then we will only have $A=1$.

~

In the second case, the only element in the periodic state will be $\ket{l}$ and then it is clear that the periodic state is a separable state. In the first case, we can rewrite the periodic state as:

\begin{equation}
\ket{\Psi^n_{l,\frac{N}{2}}}= \frac{1}{\sqrt{2}} \Big( \ket{l} + \ket{l+\frac{N}{2}} \Big) = \ket{+}\otimes \ket{l},
\end{equation}

which is also a separable state.

\end{proof}

\begin{proposition} \label{propAntiDiag2}
Let $\ket{\Psi^n_{l,r}}$ be a $n$-qubit periodic state with shift $l$ and period $r$, and let $N=2^n$. If $l+r\geq N$ then  $\ket{\Psi^n_{l,r}}$ is separable. 
\end{proposition}

\begin{proof}
Because $l+r \geq N \iff  1 \geq \frac{N-l}{r}$, one knows that the number of basis state to express $\ket{\Psi^n _{l,r}}$ will be equal to $A= \ceil*{\frac{N-l}{r}} = 1$, and thus the state is separable. 
\end{proof}

\begin{proposition}
Let $\ket{\Psi^n_{l,r}}$ be a $n$-qubit periodic state. If the period $r=2^s$ divide $N=2^n$ then for all values of the shift  respectively $l \in [\![  0, 2^s -1  ]\!]$ or $l \in  [\![  \frac{N}{2}, \frac{N}{2} + 2^s -1  ]\!]$, the state $\ket{\Psi^n_{l,2^s}}$ is a separable state, and it can be written  respectively  $\ket{\Psi^n_{l,2^s}}= \ket{+}^{\otimes (n-s)} \otimes  \ket{l}^{[s]}$ or $\ket{\Psi^n_{l,2^s}}= \ket{1}\otimes \ket{+}^{\otimes (n-s-1)} \otimes  \ket{l}^{[s]}$ (with $\ket{l}^{[s]}$ the state $\ket{l}$ written in binary notation with $s$ bits).
\end{proposition}

\begin{proof}
Let $\ket{\Psi^n_{l,r}}$ be a $n$-qubit periodic state, and let $N=2^n$. We suppose that the period divide $N$. So there exists $p$ such that $r\times p = N$, and thus there exists $(s,q)\in \mathbb{N}^2$ such that $r=2^s$ and $p=2^q$, and so $r$ is a power of 2. 

~

Therefore, the binary writting of $r$ is composed by a unique digit '1', and $n-1$ digits '0'. We know that $r=2^s$ that this '1' is at the $s$-th position. Then all multiples of $r$, which are sum of the same binary number $r$, can be written with '0' and '1' at the left of the $s$-th bit, and with only '0' bits at the right of the $s$-th bit (if we assume that the most significant bit is at the left). 

~

We focus first on the case $l \in [\![  0, 2^s -1  ]\!]$. By assuming that the shift satisfies $0 \leq l < r=2^s$, one can determine that the number of terms in the writing of $\ket{\Psi^n_{l,r}}$ is always $A= \ceil*{\frac{N-l}{r}} =p= 2^q$, and the periodic state can be written: 

\begin{equation}
\ket{\Psi^n_{l,r}} = \frac{1}{\sqrt{p}} \big( \ket{l} + \ket{l+r} + \ket{l+2r} + \cdots + \ket{l+(p-1)r} \big) ~ .
\end{equation}

By knowing that $l<r$, and that $r$ and all its multiples contain only '0' bits at the right of $s$-th bit '1' in the binary notation of $r$, we know that we can factorize, by the binary notation of $l$ in $s$ bits, the sum of basis states defining $\ket{\Psi^n_{l,r}}$. Then it will remove the $s$ less significant bits in the binary writing of $r$ and its multiples, and this lead us to the following expression:

\begin{equation}
\ket{\Psi^n_{l,r}} = \frac{1}{\sqrt{p}} \big( \ket{0} + \ket{1} + \ket{2} + \cdots + \ket{(p-1)} \big)^{[n-s]}\otimes \ket{l}^{[s]} ~ .
\end{equation}

The left factor of this tensor product is in fact the $\ket{+}^{\otimes (n-s)}$ state, and we can conclude that if $r$ divide $N$, and $l \in [\![  0, 2^s -1  ]\!]$, then we retrieve a separable state of the form: 

\begin{equation}
\ket{\Psi^n_{l,r}} = \ket{+}^{\otimes (n-s)} \otimes \ket{l}^{[s]} ~ .
\end{equation}

We focus now on the second case, i.e. when $l \in  [\![  \frac{N}{2}, \frac{N}{2} + 2^s -1  ]\!]$, which is similar the first one, excepted that we add $\frac{N}{2}$ to the shift. However, we know that $\frac{N}{2}$ is also a power of 2, and in our case, it corresponds to the state $\ket{10\dots0}$ in binary notation, when we work with $n$ bits. So this '1' bit will be present in every basis state composing the periodic state and then we can define $l'=l-\frac{N}{2}$:

\begin{equation}
\ket{\Psi^n_{l,r}} = \frac{1}{\sqrt{\dfrac{p}{2}}} \Big( \ket{l} + \ket{l+r} + \ket{l+2r} + \cdots + \ket{l+\big(\frac{p}{2}-1\big)r} \Big) ,
\end{equation}
\begin{equation}
\ket{\Psi^n_{l,r}} =  \frac{1}{\sqrt{\dfrac{p}{2}}} ~ \ket{1} \otimes \Big( \ket{l'} + \ket{l'+r} + \ket{l'+2r} + \cdots + \ket{l'+\big(\frac{p}{2}-1\big)r} \Big) ~ .
\end{equation}

The right factor of the tensor product, with the normalisation factor $\frac{1}{\sqrt{\frac{p}{2}}}$, is in fact a periodic state with shift $l' \in [\![  0, 2^s -1  ]\!]$, and a period $r=2^s$ for $(n-1)$-qubits. So by using results of the first case, we conclude that the state is separable, and we can express the whole state $\ket{\Psi^n_{l,r}}$ as:

\begin{equation}
\ket{\Psi^n_{l,r}} = \ket{1} \otimes \ket{+}^{\otimes (n-s-1)} \otimes \ket{l-\frac{N}{2}}^{[s]} ~ .
\end{equation}

\end{proof}

\begin{proposition}
Let $\ket{\Psi^n_{l,r}}$ be a $n$-qubit periodic state with shift $l$ and period $r$, and let $N=2^n$. Then there is at least $\floor*{\dfrac{N-2}{3}}+3$ pairs $(l,r)$  that define periodic states $\text{SLOCC}$ equivalents to $\ket{GHZ_n}$, and we can separate the following three cases: 

\begin{itemize}
\item the case $l=1$ and $r=1$,

\item the case $l=\frac{N}{2}-1$ and $r=1$,

\item and the $\floor*{\dfrac{N-2}{3}}+1$ other cases in the anti-diagonal defined by the relation $2l+r=N-1$.
\end{itemize}

\end{proposition}

\begin{proof}
Let $\ket{\Psi^n_{l,r}}$ be a $n$-qubit periodic state, and let $N=2^n$.

~

In the first case, if the shift is $l=1$, and the period $r=1$, we obtain the periodic state:

\begin{equation}
\ket{\Psi^n_{1,1}}= \frac{1}{\sqrt{N-1}} \sum_{x=1}^{N-1} \ket{x} = \frac{\sqrt{N}}{\sqrt{N-1}}\ket{+}^{\otimes n} - \frac{1}{\sqrt{N-1}}\ket{0}^{\otimes n} ~ .
\end{equation}

This state is a (generic) rank 2 tensor, which belongs to the smooth points of  the secant variety, and in particular is $\text{SLOCC}$ equivalent to $\ket{GHZ_n}$. 

~

In the second case, if the shift is $l=\frac{N}{2}-1$, and the period $r=1$, by remarking that $\ket{\frac{N}{2}} = \ket{10\dots 0}$ and then $\ket{\frac{N}{2}-1} = \ket{01\dots 1}$  we obtain the periodic state

\begin{equation}
\ket{\Psi^n_{\frac{N}{2}-1,1}}= \frac{1}{\sqrt{\frac{N}{2}+1}} \sum_{x=\frac{N}{2}-1}^{N-1} \ket{x} = \frac{1}{\sqrt{\frac{N}{2}+1}}    \Big( \ket{011\dots 1} +  \ket{100\dots 0} + \ket{10\dots 01} + \cdots + \ket{111\dots 1} \Big),
\end{equation}

that can be also expressed as

\begin{equation}
\ket{\Psi^n_{\frac{N}{2}-1,1}}= \frac{1}{\sqrt{\frac{N}{2}+1}}    \Big( \ket{011\dots 1} + \sqrt{\frac{N}{2}} \ket{1}\otimes \ket{+}^{\otimes n-1} \Big) ,
\end{equation}

and, this state is a also $\text{SLOCC}$ equivalent to $\ket{GHZ_n}$.

~

For the third case, we should focus on the particular anti-diagonal starting from the "point" $\{l=0,r=N-1\}$ and then moving to the top-right, by adding one to $l$ and removing two to $r$. This anti-diagonal is then defined by the following equations 

\begin{equation}
\Syst{l=k}{r=N-1-2k} ~~  \text{    with } k \in [\![  0, \frac{N}{2}-1  ]\!] ~ .
\end{equation}

We proved that the top-right point of the anti-diagonal $\{l=\frac{N}{2}-1,r=1\}$ is $\text{SLOCC}$ equivalent to $\ket{GHZ_n}$, and we can easily prove that the bottom-left point of the anti-diagonal $\{l=0,r=N-1\}$ is in fact the definition of the generalized $\ket{GHZ_n}= \frac{1}{\sqrt{2}} \big( \ket{0}^{\otimes n} + \ket{1}^{\otimes n} \big)$ state. The interesting part is to determine what is happening in the middle of the anti-diagonal. 

~

In fact, we can deduce from the anti-diagonal equations that $2l+r=N-1$. Besides, we know that the bottom-left point $\{l=0,r=N-1\}$ has only 2 basis states in its writing. Thus, we can remark that if a periodic state has only 2 basis states in its writing and if it satisfies the anti-diagonal condition $2l+r=N-1$, then state is $\text{SLOCC}$ equivalent to the $\ket{GHZ_n}$ state. In fact, this periodic state, let us call it $\ket{\Psi}$, that have only 2 basis states will be expressed as 

\begin{equation}
\ket{\Psi}= \frac{1}{\sqrt{2}}    \big( \ket{l} + \ket{l+r} \big)  ~ .
\end{equation}

But we also know, that the shift $l$ and the period $r$ of this state satisfy the condition $2l+r= l + (l+r) =N-1$. So if we work with binary notations, and by knowing that $N-1=2^n -1$ is always written with only '1' digits in its binary notation, we can conclude that the two binary numbers $l$ and $l+r$ are complementary with respect to $N=2^n$. Consequently the state $\ket{\Psi}$ is by definition an equivalent state of $\ket{GHZ_n}$. 

~

How many states are satisfying these conditions for periodic states ? In order to determine that,  we focus on the condition regarding the number of basis states in the periodic states. So we need that $A = \ceil*{\frac{N-l}{r}} = 2$, so it is equivalent to the inequation

\begin{equation}
1 <\frac{N-l}{r} \leq 2 ~ .
\end{equation}

If we substitute now $l$ and $r$ by the equations defining the anti-diagonal, we have 

\begin{equation}
1 <\frac{N-k}{N-1-2k} \leq 2 ~ .
\end{equation}

If we solve this inequation in $k$ and forget about the ceiling operation, we retrieve the result

\begin{equation}
-1 <k \leq \frac{N-2}{3} ~ .
\end{equation}

If we come back to integer numbers, then we proved that if $k \in [\![  0, \floor*{\dfrac{N-2}{3}} ]\!]$ the periodic state defined by $l=k$ and $r=N-1-2k$ has only two basis states in its writing.

~

So, the number of states equivalent to $\ket{GHZ_n}$ on the anti-diagonal is equal to $\floor*{\frac{N-2}{3}}+1$, without counting the $\{l=\frac{N}{2}-1,r=1\}$ case. In other words, from $\{l=0,r=N-1\}$ to $\{l=\floor*{\frac{N-2}{3}},r=N-1-2\floor*{\frac{N-2}{3}}\}$, all the states on the anti-diagonal are $\text{SLOCC}$ equivalent to $\ket{GHZ_n}$. 

\end{proof}

~

In this subsection, we were able to study entanglement of periodic states from a qualititve point of view. These states, which are involved in Shor's algorithm after the measurement step, show some entanglement properties that depend on the shift and the period of the considered periodic state. Indeed, we were able to point out some "rules" permitting to simplify the identification of the entanglement type of a given periodic state defined by his shift $l$ and period $r$. In the next subsection, we try to continue this work by considering periodic states after the application of the Quantum Fourier Transform, which is in fact the next step of Shor's algorithm.

\subsection{Entanglement of Periodic states after QFT}

In this subsection we investigate the entanglement of periodic states after the application of the Quantum Fourier Transform for 4-qubit systems, and we try to generalize some results for $n$-qubits periodic states. 

~

\begin{table}[!h]
 \begin{center}
  \begin{tabular}{|l||c|c|c|c|c|c|c|c|c|c|c|c|c|c|c|c|}
  \hline
   \backslashbox{$r$}{$l$} & 0 & 1 & 2 & 3 & 4 & 5 & 6 & 7 &  8 & 9 & 10 & 11 & 12 & 13 & 14 & 15   \\
   \hline
   \hline
  1 & $Gr_1$ & $G_{00cc}$ & $G_{abc0}$ & $G_{abcd}$ & $G_{abcd}$ & $G_{abcd}$ & $G_{abcd}$ & $G_{abcd}$ & $L_{a_20_{3\oplus \overline{1}}}$ &  $G_{abcd}$  & $G_{abcd}$  & $G_{abcd}$  & $G_{abcd}$  & $L_{a_20_{3\oplus \overline{1}}}$ & $L_{00c_2}$  & $Gr_1$ \\
   \hline
  2 & $Gr_1$ & $Gr_1$ & $Gr_4$& $Gr_4$ & $Gr_4$ & $Gr_4$ & $Gr_4$ & $Gr_4$ &  $Gr_4$ & $Gr_4$ & $Gr_4$ & $Gr_4$ & $Gr_4$ & $Gr_4$ & $Gr_1$ & $Gr_1$ \\
   \hline
  3 &  $G_{abcd}$ & $G_{abcd}$ & $G_{abcd}$ & $G_{abcd}$ & $G_{abcd}$ &  $G_{abcd}$ & $G_{abcd}$ & $L_{a_20_{3\oplus \overline{1}}}$ & $L_{a_20_{3\oplus \overline{1}}}$ & $L_{a_20_{3\oplus \overline{1}}}$ & $L_{00c_2}$ & $L_{00c_2}$ & $L_{00c_2}$ & $Gr_1$ & $Gr_1$ & $Gr_1$ \\
   \hline
  4 & $Gr_1$ & $Gr_1$ & $Gr_1$ & $Gr_1$ & $Gr_2$ & $Gr_2$ & $Gr_2$ & $Gr_2$ &  $Gr_2$ & $Gr_2$ & $Gr_2$ & $Gr_2$ & $Gr_1$ & $Gr_1$ & $Gr_1$ & $Gr_1$ \\
   \hline
  5 & $G_{abcd}$ & $L_{a_20_{3\oplus \overline{1}}}$ & $L_{a_20_{3\oplus \overline{1}}}$ & $L_{a_20_{3\oplus \overline{1}}}$ & $L_{a_20_{3\oplus \overline{1}}}$ & $L_{a_20_{3\oplus \overline{1}}}$ & $L_{00c_2}$ &  $L_{00c_2}$ & $L_{00c_2}$ & $L_{00c_2}$ & $L_{00c_2}$ & $Gr_1$ & $Gr_1$ & $Gr_1$ & $Gr_1$ & $Gr_1$ \\
  \hline
  6 & $Gr_4$ & $Gr_4$ & $Gr_4$ & $Gr_4$ & $Gr_4$ & $Gr_4$ & $Gr_4$ & $Gr_4$ & $Gr_4$ & $Gr_4$ & $Gr_1$ & $Gr_1$ & $Gr_1$ & $Gr_1$ & $Gr_1$ & $Gr_1$  \\  
  \hline
  7 & $L_{a_20_{3\oplus \overline{1}}}$ & $L_{a_20_{3\oplus \overline{1}}}$ & $L_{00c_2}$ & $L_{00c_2}$ & $L_{00c_2}$ & $L_{00c_2}$ & $L_{00c_2}$ & $L_{00c_2}$ &  $L_{00c_2}$ & $Gr_1$ & $Gr_1$ & $Gr_1$ & $Gr_1$ & $Gr_1$ & $Gr_1$ & $Gr_1$  \\  
  \hline
  8 & $Gr_1$ & $Gr_1$ & $Gr_1$ & $Gr_1$ & $Gr_1$ & $Gr_1$ & $Gr_1$ & $Gr_1$ &  $Gr_1$ & $Gr_1$ & $Gr_1$ & $Gr_1$ & $Gr_1$ & $Gr_1$ & $Gr_1$ & $Gr_1$  \\  
  \hline
  9 & $L_{00c_2}$ & $L_{00c_2}$ & $L_{00c_2}$ & $L_{00c_2}$ & $L_{00c_2}$ & $L_{00c_2}$ & $L_{00c_2}$ & $Gr_1$ &  $Gr_1$ & $Gr_1$ & $Gr_1$ & $Gr_1$ & $Gr_1$ & $Gr_1$ & $Gr_1$ & $Gr_1$  \\  
  \hline
  10 & $Gr_4$ & $Gr_4$ & $Gr_4$ & $Gr_4$ & $Gr_4$ & $Gr_4$ & $Gr_1$ & $Gr_1$ &  $Gr_1$ & $Gr_1$ & $Gr_1$ & $Gr_1$ & $Gr_1$ & $Gr_1$ & $Gr_1$ & $Gr_1$ \\  
  \hline
  11 & $L_{00c_2}$ & $L_{00c_2}$ & $L_{00c_2}$ & $L_{00c_2}$ & $L_{00c_2}$ & $Gr_1$ & $Gr_1$ & $Gr_1$ &  $Gr_1$ & $Gr_1$ & $Gr_1$ & $Gr_1$ & $Gr_1$ & $Gr_1$ & $Gr_1$ & $Gr_1$  \\  
  \hline
  12 & $Gr_2$ & $Gr_2$ & $Gr_2$ & $Gr_2$ & $Gr_1$ & $Gr_1$ & $Gr_1$ & $Gr_1$ &  $Gr_1$ & $Gr_1$ & $Gr_1$ & $Gr_1$ & $Gr_1$ & $Gr_1$ & $Gr_1$ & $Gr_1$  \\  
  \hline
  13 & $L_{00c_2}$ & $L_{00c_2}$ & $L_{00c_2}$ & $Gr_1$ & $Gr_1$ & $Gr_1$ & $Gr_1$ & $Gr_1$ & $Gr_1$ & $Gr_1$ & $Gr_1$ & $Gr_1$ & $Gr_1$ & $Gr_1$ & $Gr_1$ & $Gr_1$  \\  
  \hline
  14 & $Gr_4$ & $Gr_4$ & $Gr_1$ & $Gr_1$ & $Gr_1$ & $Gr_1$ & $Gr_1$ & $Gr_1$ & $Gr_1$ & $Gr_1$ & $Gr_1$ & $Gr_1$ & $Gr_1$ & $Gr_1$ & $Gr_1$ & $Gr_1$  \\  
  \hline
  15 & $L_{00c_2}$ & $Gr_1$ & $Gr_1$ & $Gr_1$ & $Gr_1$ & $Gr_1$ & $Gr_1$ & $Gr_1$ & $Gr_1$ & $Gr_1$ & $Gr_1$ & $Gr_1$ & $Gr_1$ & $Gr_1$ & $Gr_1$ & $Gr_1$ \\  
  \hline
  \end{tabular}
\caption{Verstaete \textit{et al.} families of the resulting states after applying the QFT on periodic states depending on their shift $l$ and period $r$. }\label{periodicAfterQFT4res}
 \end{center}

\end{table}

From the results presented in this Table \ref{periodicAfterQFT4res}, as we did for the case before the QFT, we can extract the following properties for 4-qubit periodic states after the Quantum Fourier Transform:

\begin{itemize}
\item When $r=8$, the state always belong to the $Gr_1$ orbit (separable),

\item For all the states that have the shift $l$ and period $r$ satisfying $l+r \geq N$ (i.e. we are on or under the anti-diagonal from $\{l=1,r=15\}$ and $\{l=15,r=1\}$), they belong to the $Gr_1$ orbit of the nullcone, and thus are separable states,

\item If the period takes the values $r=1$, $r=2$, $r=4$ and $r=8$, and if the shift $0\leq l \leq r-1$, then the state is a separable state,

\item If the period take the values $r=1$, $r=2$, and $r=4$, and if the shift $\frac{N}{2} \leq l \leq \frac{N}{2}+ r-1$, then the state is not a separable state,

\item If the period is equal to $r=2$, then if the shift $2 \leq l \leq N-3$, then the state belongs to the $Gr_4$ orbit, and thus is not a separable state.

\end{itemize}

~

We were not able, as it was the case for periodic states before the QFT, to generalize all these observations to the general $n$-qubit case. Most of these behaviors were observed for both 3-qubit (see Appendix \ref{appendixC}) and 4-qubit case, but such a table cannot be constructed for the case of 5-qubits as there is no known {SLOCC} classification.

~

We were able, however, to propose one result with respect to periodic states after QFT that are separables, and we use for that a basic property of the Quantum Fourier Transform:

\begin{proposition} \label{propAntiDiagAfter}
Let $\ket{\Psi^n_{l,r}}$ be a $n$-qubit periodic state with shift $l$ and period $r$, and let $N=2^n$. Then, after the application of the QFT, all the states on and under the anti-diagonal, defined from the cell $\{l=1,r=N-1\}$ to $\{l=N-1,r=1\}$, are separable states. 
\end{proposition}

\begin{proof}
We have demonstrated in the Proposition \ref{propAntiDiag2} that the number of states in the writting of the periodic states $\ket{\Psi^n_{l,r}}$ with $l+r \geq N$ (under or on the anti-diagonal) is one, and thus they are computational basis states. By using the Remark \ref{QFTbasisState}, we can conclude that after the application of QFT, we always will retrieve a factorized and separable state.
\end{proof}

\subsection{Entaglement of Periodic states through QFT}

In this subsection, we use the 4-qubit hyperdeterminant $|Det_{2222}|$ as a quantitative measure of entanglement to investigate the amount of entanglement after each gate composing the Quantum Fourier Transform. We focus on the case when only 4-qubit periodic states are input of the QFT. 

~

\begin{figure}[!h]
   \begin{align*}
   \Qcircuit @C=1em @R=.7em {
   \lstick{\ket{x_1}}  & \gate{H}   & \gate{R_2}  & \gate{R_3} & \gate{R_4} & \qw          & \qw              & \qw             & \qw         & \qw              & \qw & \qw & \multigate{3}{SWAP} \\
   \lstick{\ket{x_2}}  & \qw           & \ctrl{-1}        & \qw             & \qw             & \gate{H}  & \gate{R_2} & \gate{R_3} & \qw         & \qw              & \qw & \qw & \ghost{SWAP} \\
   \lstick{\ket{x_3}}  & \qw           & \qw               & \ctrl{-2}      & \qw             & \qw          & \ctrl{-1}       & \qw             & \gate{H} & \gate{R_2} & \qw & \qw &  \ghost{SWAP}\\
   \lstick{\ket{x_4}}  & \qw           &\qw                & \qw             & \ctrl{-3}      & \qw          & \qw              & \ctrl{-2}      & \qw         & \ctrl{-1}       & \gate{H} & \qw & \ghost{SWAP} \\
   }
   \end{align*}
   \caption{Eleven gates composing the quantum circuit for 4-qubit Quantum Fourier Transform.}
   \label{QFT_steps}
   \end{figure}
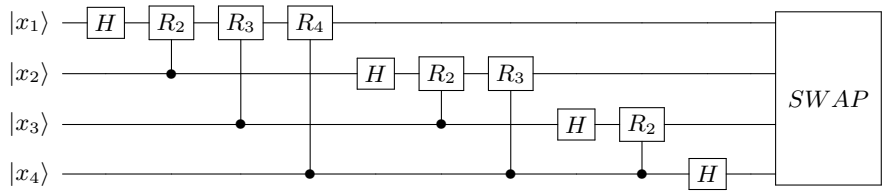

~

In fact, when we look at the absolute value of $Det_{2222}$ after each unitary gate composing QFT, when applied on a 4-qubit periodic state, we can distinguish three different cases:
\begin{enumerate}
\item The value of hyperdeterminant is never equal to zero before, after and during QFT,
\item The value of hyperdeterminant is zero before QFT, and then become non-null after a particular step of QFT and does not vanish through next gates,
\item The value of hyperdeterminant is always equal to zero before, after and during QFT.
\end{enumerate}

~

The first case only happen when we have as input the couples $(1,3)$ and $(2,3)$, with $(l,r)$ representing a given periodic state with $l$ the shift and $r$ the period. In both cases, the periodic state generated belong to the family $G_{abc0}$, which does not annihilate the 4-qubit hyperdeterminant. After applying the QFT to these states, we retrieve states belonging to $G_{abcd}$ generic family. We plot in Figure \ref{hyper_shor_1} the evolution of $|Det_{2222}|$ throughout the gates composing the QFT.  For instance, at Step 0 we retrieve the input periodic state ; at Step 1 we apply Hadamard gate to the initial state ; at Step 3 we retrieve the state resulting of the application of first Hadamard, the first c-$R_2$ and the first c-$R_3$ to the intial periodic state ; and so on. We observe a global decreasing of the absolute value of the hyperdeterminant while we move to a more generic Verstraete \textit{et al.} family.

\begin{figure}[!h]
\begin{center}
  \includegraphics[width = 12cm]{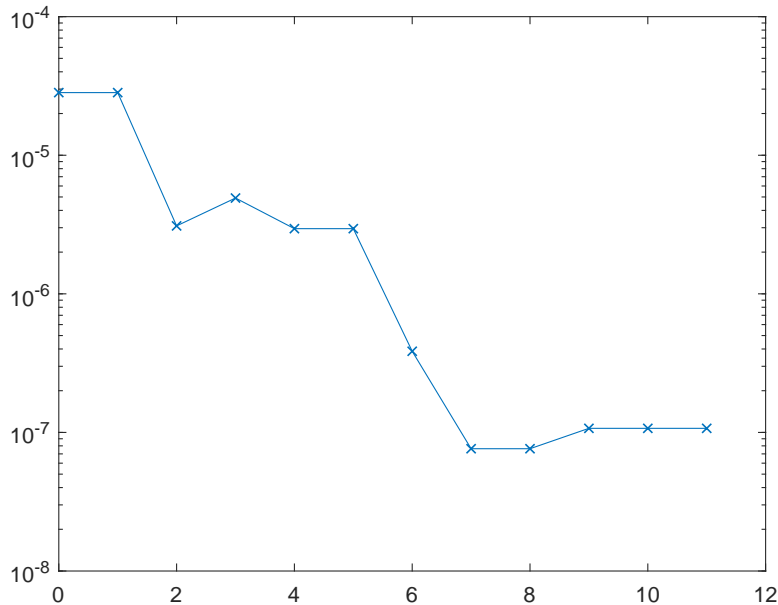}
  \caption{Evolution of the absolute value of hyperdeterminant in function of the QFT steps, for the periodic state with $(l,r)=(1,3)$.}\label{hyper_shor_1}
\end{center}
\end{figure}

~

The second case happen for couples $(l,r) \in \mathcal{S}$, with $\mathcal{S}$ the set defined bellow in Equation (\ref{equ_set_second_case}) :
\begin{equation}\label{equ_set_second_case}
\mathcal{S}= \{(0,3),(0,5),(2,1),(3,1),(3,3),(4,1),(4,3),(5,1),(5,3),(6,1),(6,3),(7,1),(9,1),(10,1),(11,1),(12,1)   \} ~ .
\end{equation}

We start with states in the dual variety $X^*$, where $Det_{2222}$ vanishes. Then depending on the periodic state, the value of the hyperdeterminant will become non-null after one of the c-$R_k$ gate. We can distinguish different types of behavior concerning the evolution of $|Det_{2222}|$ throughout QFT gates, represented in Figure \ref{hyper_shor_2}.

\begin{figure}[!h]
\begin{center}
  \centerline{\includegraphics[width = 10.5cm]{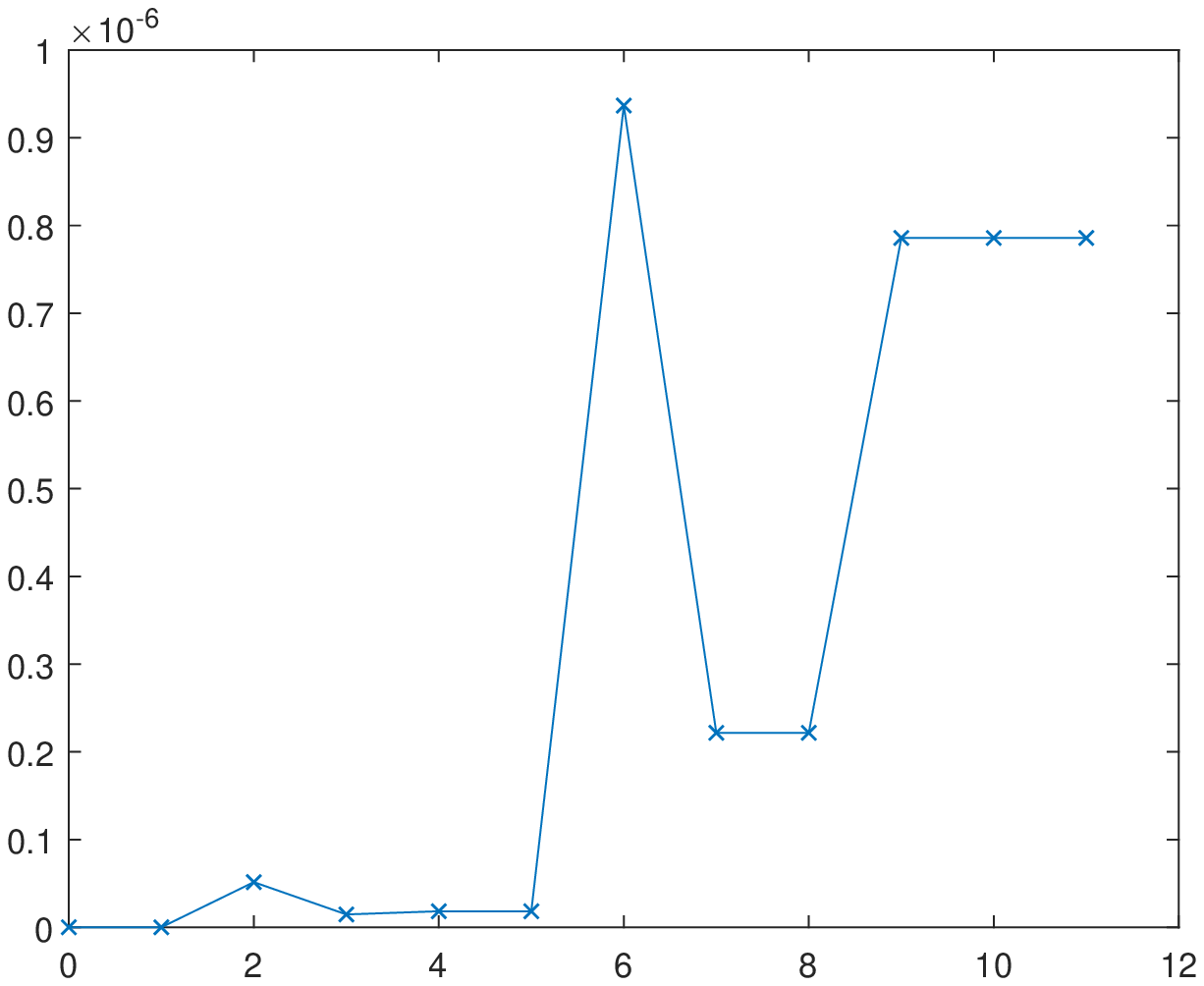}\includegraphics[width = 10.5cm]{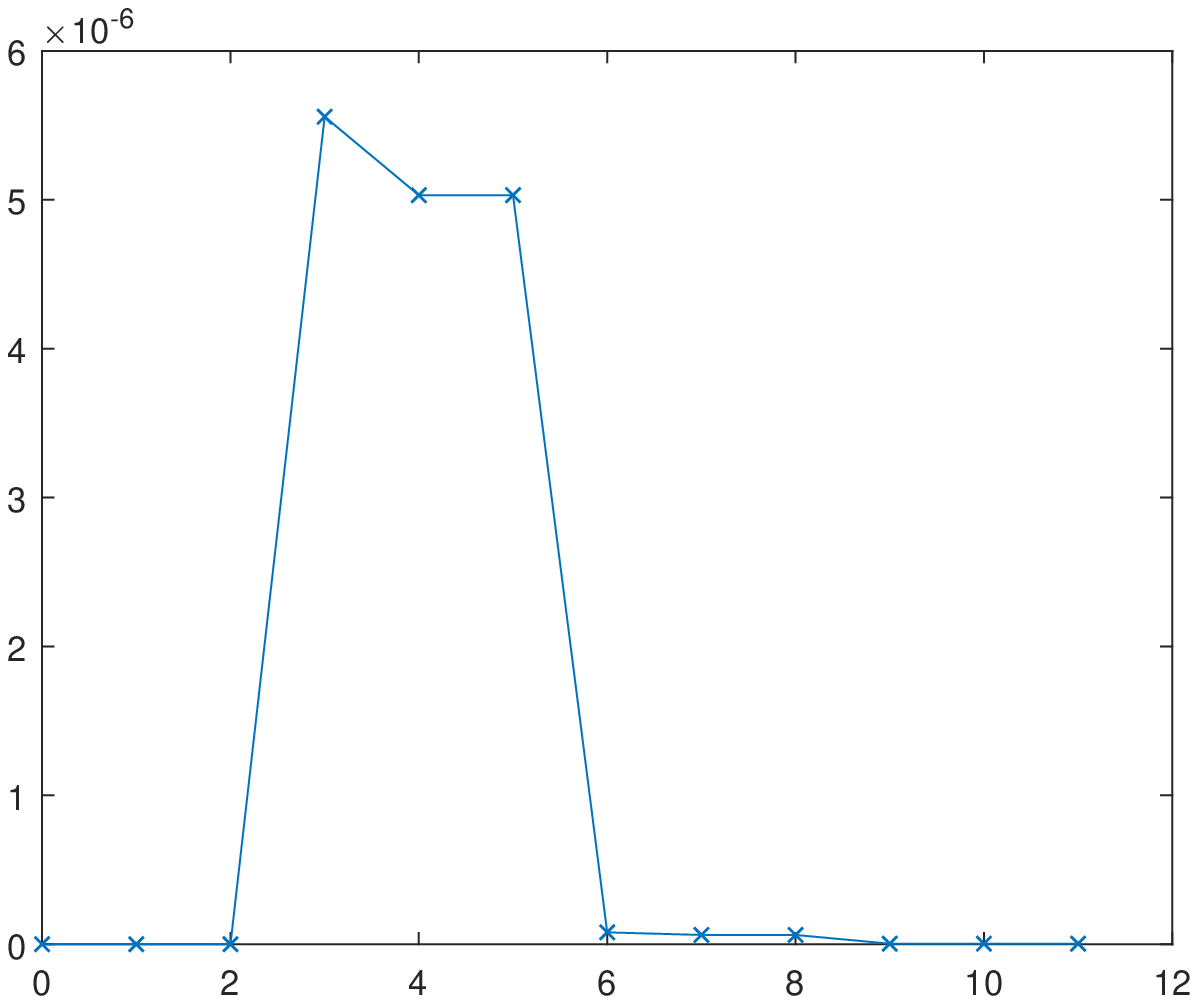}}
  \centerline{\includegraphics[width = 10.5cm]{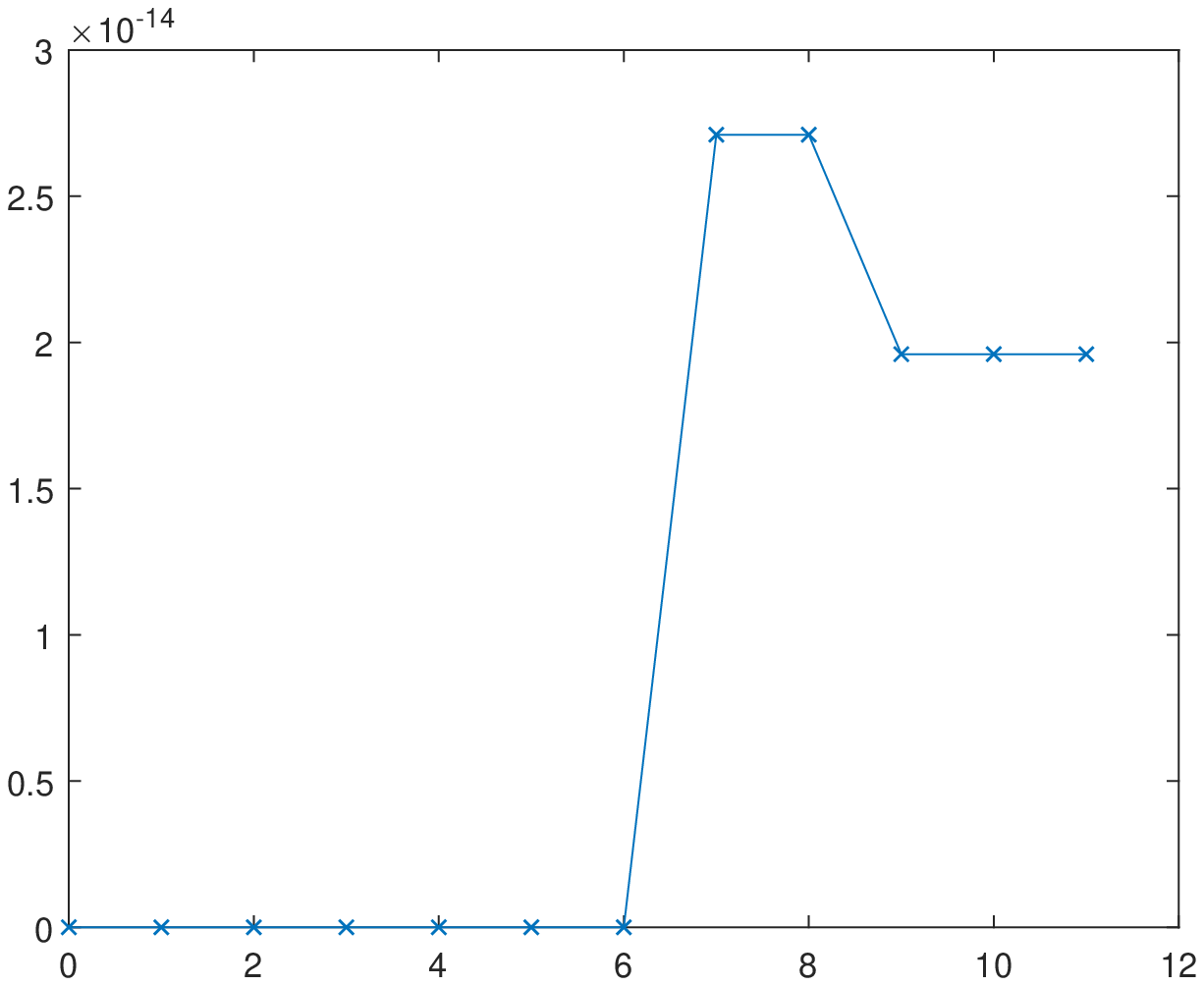}\includegraphics[width = 10.5cm]{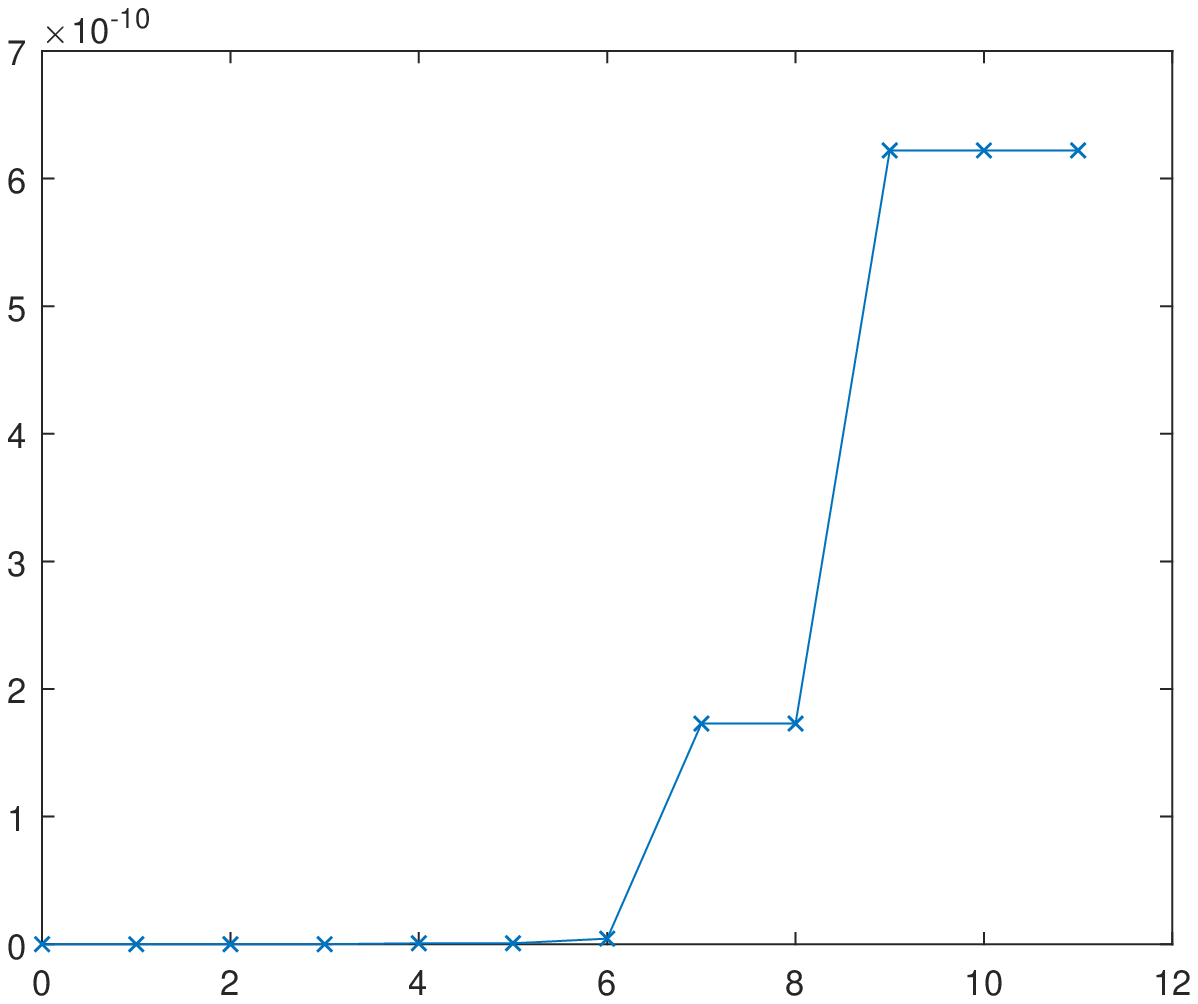}}
  \caption{Evolution of the absolute value of hyperdeterminant in function of the QFT steps, for periodic states with $(l,r)=(0,3)$ (top left), $(l,r)=(5,3)$ (top right), $(l,r)=(11,1)$ (bottom left) and $(l,r)=(9,1)$ (bottom right).}\label{hyper_shor_2}
\end{center}
\end{figure}

~

The third case (when the value of hyperdeterminant is always zero during the QFT) happen with all remaining couples $(l,r) \not\in \mathcal{S} \cup \{(1,3),(2,3)\}$. It does not necessarily mean that QFT does not modify entanglement class and type of these periodic states, but it means that all gates composing the QFT does not bring these states out of the dual variety. 

~

Besides, we remark that the value of the hyperdeterminant does not change after the application of Hadamard gates (first, fifth, eighth and tenth gate) and the last Swap gate, as expected ($Det_{2222}$ is invariant under local unitary operation and permutation of qubits). In fact, changes only appear after applying the c-$R_k$ gates, and thus they are responsible of generation or modification of entanglement by the Quantum Fourier Transform.

~

Moreover, for most of the 4-qubit periodic states (92.5\%), the hyperdeterminant measure does not change after or during the QFT (staying at zero), and this has already been pointed out by Shimoni \textit{et al.} in \cite{Shimoni} for the Groverian measure of entanglement. One can also remark that the hyperdeterminant is always null and never changes for a couple $(l,r)$ when $r$ is a power of 2, as it was mentioned by Kendon and Munro in \cite{Kendon}. However, we can have a qualitative change of entanglement, as it has been seen in the previous sections. We also observe that the only values of $r$ for which $Det_{2222}$ changes is $r=1$, $r=3$ and $r=5$.

\subsection{Entanglement and Quantum Fourier Transform} \label{3qubit_shor_eett}

We now no more consider  the specific case of Shor's algorithm, to come back to general quantum states. In this subsection, we will briefly discuss the influence of applying the QFT on the entanglement nature of quantum systems. It is well known that Quantum Fourier Transform can "create" or modify the entanglement when it is applied in the general case, but we yet do not know how this is actually happening. We will thus try to give some elements of answer.

\subsubsection{Linear Shift Invariant property}

One of the most important property of the Quantum Fourier Transform is what is called the \textit{Linear Shift Invariant} property. It is known that if a $n$-qubit state shows some periodic behavior with a shift, after the application of QFT, we will retrieve a non-shifted state (a periodic state starting with the basis state $\ket{00\dots 0}$) with some periodic behavior directly related to the previous period and the number of basis states $Q=2^n$.

~

In their work, Most \textit{et al.} investigated this property to deduce an approximative description of periodic quantum states after QFT. We wanted to quote precisely the authors \cite{Most} : 

"In analogy to the discrete Fourier transform (DFT), the QFT is used in order to reveal periodicities in its input. In particular, the amplitudes of the state $\ket{\Psi^n_{r,l}}$  make out a periodic series, and when the DFT is applied to it, the resulting series can be approximated by a periodic series of the same sort, that is, one in which the indices of the nonzero terms make out an arithmetic progression. In the resulting series, though, the common difference is $Q/r$, the initial term is zero, and additional phases are added. [...] Since applying the QFT to a quantum state is equivalent to applying the DFT to its amplitudes, the action of the QFT on periodic states can be approximately described as: $\ket{\Psi^n_{r,l}} \xrightarrow{QFT} \ket{\Psi^n_{Q/r,0}}$".

~

While this approximation may be needed to compute more easly some numerical measure of entanglement, we believe that this numerical  measure does not preserve the entanglement nature of periodic states after QFT in terms of SLOCC-orbit. We would like to illustrate this, with a basic. Let us consider the following periodic state:

\begin{equation}
 \ket{\Psi^3_{2,2}}= \frac{1}{\sqrt{3}} \big( \ket{010} + \ket{100} + \ket{110} \big) = \frac{1}{\sqrt{3}} \Big( \ket{01} + \ket{10} + \ket{11} \Big)\otimes \ket{0}  ~ .
\end{equation}

This state certainly belongs to the $\mathcal{O}_4$ orbit of the 3-qubits classification (see Fig. \ref{3onion}). When we apply the Quantum Fourier Transform to this state, we retrieve the state $ \ket{\Psi^*}$:

\begin{equation} 
\ket{\Psi^*} = QFT\ket{\Psi^3_{2,2}}= \frac{1}{\sqrt{8\times 3}} \Big( 3\ket{000} - \ket{001} - \ket{010} - \ket{011} + 3\ket{100} -  \ket{101} - \ket{110} -  \ket{111}  \Big)  ~ .
\end{equation}

This last state show in fact some periodic properties, and it can be written as \\$\ket{\Psi^*} = \alpha \ket{+}^{\otimes 3} + \beta \ket{000} + \beta \ket{100}$ $= \alpha \ket{+}^{\otimes 3} + \delta \ket{+}\ket{00}$ $= \ket{+}\otimes \big(\alpha \ket{+}^{\otimes 2} + \delta\ket{0}^{\otimes 2} \big)$, and thus belongs to the $\mathcal{O}_3$ orbit. 

~

However, if we focus on the periodic state with the shift $l=0$, and the period $r=\frac{N}{2} = \frac{8}{2} = 4$, we obtain the following state

\begin{equation}
 \ket{\Psi^3_{4,0}}= \frac{1}{\sqrt{2}} \big( \ket{000} + \ket{100} \big) = \frac{1}{\sqrt{2}}\ket{+} \otimes \ket{00},
\end{equation}

which is a separable state, and it is not $\text{SLOCC}$ equivalent to the state $\ket{\Psi^*}$. Therefore, from a qualitative point of view, the approximation used in \cite{Most} cannot be used. However, if we compute the absolute value of the 3-qubit Cayley hyperdeterminant and consider it as a measure of entanglement, we retrieve the same value (which is zero) for both $\ket{\Psi^*}$ and $\ket{\Psi^3_{4,0}}$ states.

\subsubsection{Other remarks on Quantum Fourier Transform}

Usually, when the Quantum Fourier Transform is defined, some of its properties are also mentioned. One basic property of QFT is that it sends the basis state $\ket{0}^{\otimes n}$ to the state $\ket{+}^{\otimes n}$, and so it has the property of building fully parallelized states. This property can be generalized to any basis state:

\begin{rem} \label{QFTbasisState}
If we apply the Quantum Fourier Transform to one of the computational basis states, then we always retrieve a separable state. We can in fact directly deduce this from the \textbf{product representation} recalled in equation (5.4)  of section II.5 of the book \cite{QCQI}. For a basis state $\ket{j_1 j_2 \cdots j_n}$, we retrieve the factorized state $\dfrac{1}{\sqrt{2^n}} (\ket{0} + e^{2i\pi 0 \cdot j_n}\ket{1}) \cdot (\ket{0} + e^{2i\pi 0 \cdot j_{n-1}j_n}\ket{1}) \cdots (\ket{0} + e^{2i\pi 0 \cdot j_1j_2\dots j_n}\ket{1}) $ after the application of the Quantum Fourier Transform.
\end{rem}

 ~

But this property is also related to entanglement, since one can deduce that any basis state stay in the set of separable states after the application of Quantum Fourier Transform. However, it is not true for all separable states, as we can see in Table \ref{periodic4res} and Table \ref{periodicAfterQFT4res}, and at the end of this section. 

~

Tables \ref{periodic4res}, \ref{periodicAfterQFT4res} tell us that $QFT$ can transform separable states to entangled one and can change the nature of entanglement. Let us propoose another simple but clear example for 3-qubits systems. Let us apply the QFT to the well known state $\ket{W}=\frac{1}{\sqrt{3}}\big(\ket{001}+\ket{010}+\ket{100}\big)$. We know that the $\ket{W}$ state belongs to the $\mathcal{O}_5$ orbit for the 3-qubits. After QFT we obtain the following state : 

\begin{equation}\begin{split} TFQ_8 \ket{W} = \frac{1}{\sqrt{8}}\frac{1}{\sqrt{3}}\Big(3\ket{000} + (\omega + \omega^2 + \omega^4)\ket{001} + \omega^2 \ket{010}+ (\omega^3+\omega^6+\omega^4)\ket{011}  \\+ \ket{100} + (\omega^5 + \omega^2 + \omega^4)\ket{101} + \omega^6\ket{110}+ (\omega^7 + \omega^6 + \omega^4)\ket{111} \Big) \end{split},\end{equation}
with $\omega = e^{\frac{2i\pi}{8}}$.

~

By computing the Cayley hyperdeterminant, one can verify that it is equal to $\dfrac{-i}{36}$, which means that this state belongs to the $\mathcal{O}_6$ orbit (the $\ket{GHZ}$ orbit). Therefore, we modified the entanglement class of the state by applying the QFT.

~

It can thus be interesting to look at the equivalence classes related to the action of the group $\text{SLOCC} \cup \{QFT,QFT^{-1}\}$ or $\text{SLOCC} \cup \{H,c \text{-} R_k, SWAP\}$ (where $c-Rk$ is the controled $R_k$ gate, see Figure \ref{QFT_steps}). It could be interesting to study entanglement generated by circuit of $H$, $SWAP$, and $c \text{-} R_k$ gates (composing $QFT$), as it was done by Bataille \textit{et al.} in \cite{bataille} for $c \text{-}Z$ and $SWAP$ gates, to understand more deeply the influence of such gates on entanglement. 

~

As a first step in that direction we choose to focus on the case of 2-qubits and 3-qubits to study this last question. In the case of 2-qubit systems, we can find that there is only one equivalence class under the action of $G'=\text{SLOCC} \cup \{QFT,QFT^{-1}\}$. In fact, if we take the state $\ket{\Psi_1}$ defined by 

\begin{equation}
\ket{\Psi_1} = \frac{1}{\sqrt{2}} \Big( \ket{00} + \ket{01} \Big) ~ .
\end{equation}

When we apply the Quantum Fourier Transform to this state, we retrieve the entangled state 

\begin{equation}
\ket{\Psi_2} = \frac{1}{2\sqrt{2}} \Big( 2\ket{00} + (1+i)\ket{01} + (1-i)\ket{11} \Big) ~ .
\end{equation}

We know that for 2-qubit systems, there is only two $\text{SLOCC}$ entanglement classes : separable or entangled (EPR). Since we can move from the separable state $\ket{\Psi_1}$ to the entangled state $\ket{\Psi_2}$, then there is only one orbit under the action of $G'$ . 

~

In order to investigate the 3-qubit case, let us compute the QFT for several examples of 3-qubit states: 

\begin{equation}
\ket{\Phi_1} = \frac{1}{\sqrt{3}} \Big( \ket{001}+\ket{010} + \ket{100} \Big) \in \mathcal{O}_5 ,
\end{equation}
\begin{equation} \begin{split}
\ket{\Phi_2} = \frac{1}{\sqrt{24}} \Big( 3\ket{000} + (\omega + \omega^2 + \omega^4)\ket{001} + \omega^2 \ket{010}+ (\omega^3+\omega^6+\omega^4)\ket{011}  \\+ \ket{100} + (\omega^5 + \omega^2 + \omega^4)\ket{101} + \omega^6\ket{110}+ (\omega^7 + \omega^6 + \omega^4)\ket{111} \Big) \in \mathcal{O}_6 
 \end{split},\end{equation}
\begin{equation}
\ket{\Phi_3} = \frac{1}{\sqrt{3}} \Big( \ket{100}+\ket{110} + \ket{111} \Big) \in \mathcal{O}_3 ,
\end{equation}
\begin{equation}
\ket{\Phi_4} = \frac{1}{\sqrt{2}} \Big( \ket{110} + \ket{111} \Big) \in \mathcal{O}_4 ,
\end{equation}
\begin{equation} \begin{split}
\ket{\Phi_5} = \frac{1}{4} \Big( 2\ket{000} + (\omega^7 - i)\ket{001} - (1+i) \ket{010}+ (i+\omega^5)\ket{011} \\  + (\omega^3 - i)\ket{101} + (i-1)\ket{110}+ (\omega + i)\ket{111} \Big) \in \mathcal{O}_6 
 \end{split},\end{equation}
\begin{equation}
\ket{\Phi_6} = \frac{1}{\sqrt{2}} \Big( \ket{001} + \ket{010} \Big) \in \mathcal{O}_3 ,
\end{equation}
\begin{equation}
\ket{\Phi_7} = \frac{1}{\sqrt{2}} \Big( \ket{101} + \ket{111} \Big) \in \mathcal{O}_1 ,
\end{equation}
\begin{equation} 
\ket{\Phi_8} = \frac{1}{4} \Big( 2\ket{000} + ( \omega^5 + \omega^7)\ket{001} + (\omega^5+\omega^7)\ket{011}  - 2 \ket{100} + (\omega + \omega^3) \ket{101} + (\omega + \omega^3 )\ket{111} \Big) \in \mathcal{O}_3 ,
\end{equation}
\begin{equation}
\ket{\Phi_9} = \frac{1}{\sqrt{2}} \Big( \ket{001} + \ket{011} \Big) \in \mathcal{O}_1 ,
\end{equation}
\begin{equation}
\ket{\Phi_{10}} = \frac{1}{\sqrt{2}} \Big( \ket{000} + \ket{101} \Big) \in \mathcal{O}_2 ,
\end{equation}
\begin{equation} \begin{split}
\ket{\Phi_{11}} = \frac{1}{4} \Big( 2\ket{000} + (1 + \omega^5 )\ket{001} + (1+i) \ket{010}+ (1+\omega^7)\ket{011} \\ + (1+ \omega)\ket{101} + (1 - i )\ket{110}+ (1 + \omega^3 )\ket{111} \Big) \in \mathcal{O}_6 
 \end{split},\end{equation}
\begin{equation}
\ket{\Phi_{12}} = \frac{1}{\sqrt{2}} \Big( \ket{000} + \ket{011} \Big) \in \mathcal{O}_3 ,
\end{equation}
\begin{equation}
\ket{\Phi_{13}} = \frac{1}{\sqrt{2}} \Big( \ket{000} + \ket{110} \Big) \in \mathcal{O}_4 ,
\end{equation}
\begin{equation} 
\ket{\Phi_{14}} = \frac{1}{4} \Big( 2\ket{000} + (1 - i )\ket{001} + (1 + i) \ket{011} + 2\ket{100} + (1-i)\ket{101} + (1 + i )\ket{111} \Big) \in \mathcal{O}_3 ,
\end{equation}
\begin{equation}
\ket{\Phi_{15}} = \frac{1}{\sqrt{2}} \Big( \ket{000} + \ket{010} \Big) \in \mathcal{O}_1  ~ .
\end{equation}

~

These states are related to each other, using the QFT, as detailled bellow

$$ \ket{\Phi_1} \xrightarrow[]{QFT} \ket{\Phi_{2}} \xrightarrow[]{QFT} \ket{\Phi_{3}} ~~, ~~~~ \ket{\Phi_4} \xrightarrow[]{QFT} \ket{\Phi_{5}} \xrightarrow[]{QFT} \ket{\Phi_{6}}  ~~, ~~~~  \ket{\Phi_7} \xrightarrow[]{QFT} \ket{\Phi_{8}} \xrightarrow[]{QFT} \ket{\Phi_{9}}, $$
$$ \ket{\Phi_{10}} \xrightarrow[]{QFT} \ket{\Phi_{11}} \xrightarrow[]{QFT} \ket{\Phi_{12}} ~~, ~~~~ \ket{\Phi_{13}} \xrightarrow[]{QFT} \ket{\Phi_{14}} \xrightarrow[]{QFT} \ket{\Phi_{15}}  ~ .$$

~

By using the visual representation proposed in Figure \ref{graph_slocc_qft}, one can verify that, starting from a specific state $\ket{\Phi_{i}}$, we can reach any state in any orbit, by applying a succession of operations in $\text{SLOCC}$ group and/or the QFT (and its inverse). 

  \begin{figure}[!h]
\begin{center}
  \includegraphics[width = 12cm]{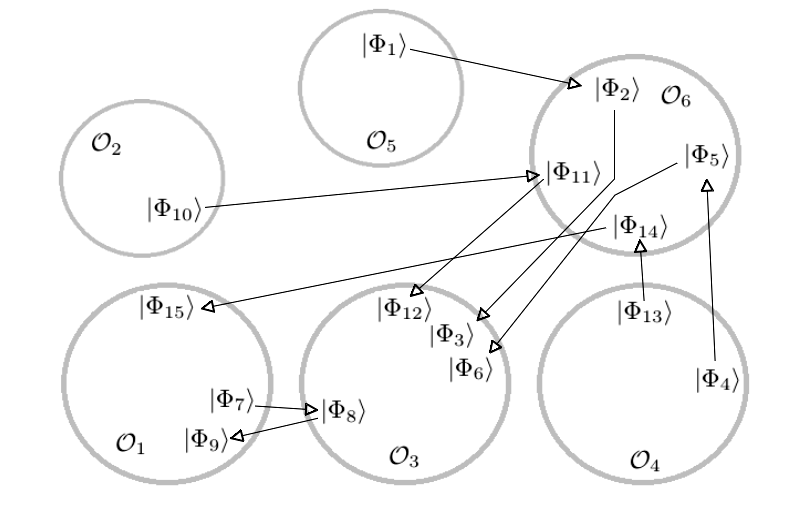}
  \caption{Quantum states in the same grey circle belong to the same $\text{SLOCC}$ orbit. White-headed arrows correspond to the application of the Quantum Fourier Transform (the state close to the arrow's head is the result of the application of QFT to the state at the arrow's root).}\label{graph_slocc_qft}
\end{center}
\end{figure}

~

We can thus conclude that there is only one orbit under the action of $G'$, both in 2-qubits 3-qubits cases. It can be more difficult to investigate the 4-qubit case, because there is an infinite number of $\text{SLOCC}$ orbit. However, we can first tackle the problem by only considering Verstraete \textit{et al.} families of sub-families, and the 9 orbit of the nullcone. One can use for that the results presented in Table \ref{periodic4res} and Table \ref{periodicAfterQFT4res} to extract 4-qubit states equivalents by Quantum Fourier Transform.

%

\section{Conclusion}\label{conclusion}

In this work, we were able to produce a detailed study of the  entanglement classes emerging during Grover's and Shor's algorithm, in the case of 4-qubit systems. For Grover's algorithm, we determined what type of entanglement is generated depending on the marked elements and their number. For Shor's algorithm, we focused on periodic states generated after the first measure in the period-finding algorithm. We presented entanglement families of periodic states depending on their parameters (shift and period), and how these classes change after the application of the Quantum Fourier Transform. Some of these results were generalized to any $n$-qubit version of the algorithm. 

Our analysis shows that some $4$-qubit states never shows up in both algorithms. It is the case of the famous $\ket{W}$ state (as it was observed in \cite{HJN} for Grover's algorithm on some tripartite systems for the standard regime) but also $\ket{L_{abc_2}}$ or $\ket{L_{ab_3}}$. So far we do not have any (geometric) explanations of this phenomena.

Regarding the $QFT$ and its influence on entanglement we have shown that for $n=2$ and $n=3$ qubits, we have shown that all quantum states are equivalent up to a sequence of transformation of the set $\SLOCC \cup\{QFT\}$.  
\section*{Acknowledgments}\label{remerciements}

This work is supported by the R\'egion Bourgogne Franche-Comt\'e, project PHYFA (contract 20174-06235) and  the French “Investissements d’Avenir” programme, project ISITE-BFC (contract ANR-15-IDEX-03).

\newpage

\appendix
\section{Examples of orbits related with sets of marked elements}\label{appendixA}
The following Table provide examples of sets of marked elements which allows to reach the corresponding orbits and Verstraete families by running Grover's algorithm in the 4-qubit case.
\begin{table}[!h]
 \begin{center}
  \begin{tabular}{|c|c|}
  \hline
   Orbit & $~~~S~~~$  \\
   \hline
  $G_{abcd}$ &  $\{ \ket{0000}, \ket{0001}, \ket{0010} ,\ket{0101}, \ket{1010}, \ket{1111}  \}$   \\
   \hline
  $G_{abc0}$ & $\{\ket{0000},\ket{1111}\}$ \\
   \hline
  $G_{00cc}$ & $\{\ket{0000}\}$ \\
   \hline
  $G_{a000}$ & $\{  \ket{0000}, \ket{0011}, \ket{1100}, \ket{1111} \}$ \\
   \hline
  $G_{ab00}$ & $\{  \ket{0000}, \ket{0011}, \ket{1101}, \ket{1110} \}$ \\
   \hline
  $L_{abc_2}$ & $\emptyset$  \\  
  \hline
  $L_{00c_2}$ & $\{\ket{0000},\ket{0011}\}$  \\  
  \hline
  $L_{aa0_2}$ & $\{\ket{0000},\ket{0101}\}$  \\  
  \hline
  $L_{a00_2}$ & $\{\ket{0000}, \ket{0110}, \ket{1001}, \ket{1111} \}$  \\  
  \hline
  $L_{ab0_2}$ & $\{ \ket{0000}, \ket{0001}, \ket{0010}, \ket{0101}, \ket{1010} \}$ \\  
  \hline
  $L_{a_2b_2}$ & $\{ \ket{0000}, \ket{0001}, \ket{0010}, \ket{0100}, \ket{1001}  \}$  \\  
  \hline
  $L_{0_2b_2}$ & $\{ \ket{0000}, \ket{0001}, \ket{0110}  \}$  \\  
  \hline
  $L_{ab_3}$ & $\emptyset$  \\  
  \hline
  $L_{0b_3}$ & $\emptyset$  \\  
  \hline
  $L_{a0_3}$ & $\emptyset$  \\  
  \hline
  $L_{a_4}$ & $\{ \ket{0000}, \ket{0001}, \ket{0010}, \ket{0101}, \ket{0110}, \ket{1101} \}$  \\  
  \hline
  $L_{a_20_{3\oplus \overline{1}}}$ & $\{ \ket{0000}, \ket{0001}, \ket{1110}  \}$  \\  
  \hline
  $Gr_8$ & $\{\ket{0000},\ket{0111}\}$  \\  
  \hline
  $Gr_7$ & $\{  \ket{0000}, \ket{0001}, \ket{0110}, \ket{1011} \}$  \\  
  \hline
  $Gr_6$ & $\{  \ket{0000}, \ket{0001}, \ket{0010}, \ket{1100} \}$  \\  
  \hline
  $Gr_5$ & $\{ \ket{0000}, \ket{0011}, \ket{0101}, \ket{1001}  \}$ \\  
  \hline
  $Gr_4$ & $\{\ket{0000},\ket{0001}\}$  \\  
  \hline
  $Gr_3$ & $\{  \ket{0000}, \ket{0001}, \ket{0010}, \ket{0100} \}$ \\  
  \hline
  $Gr_2$ & $\{  \ket{0000}, \ket{0001}, \ket{0110}, \ket{0111}  \}$  \\  
  \hline
  $Gr_1$ & $\{  \ket{0000}, \ket{0001}, \ket{0010}, \ket{0011} \}$  \\  
  \hline
  \end{tabular}
\caption{Examples of family of marked elements $S$ and the corresponding family or orbit reached by the algorithm in the $2\times 2\times 2 \times 2$ case.}\label{2222res}
 \end{center}

\end{table}

\newpage

\section{3-qubits periodic states}\label{appendixC}

\begin{table}[!h]
 \begin{center}
  \begin{tabular}{|l||c|c|c|c|c|c|c|c|}
  \hline
   \backslashbox{$r$}{$l$} & 0 & 1 & 2 & 3 & 4 & 5 & 6 & 7    \\
   \hline
   \hline
  1 & $\calO_1$ & $\calO_6$ & $\calO_4$ & $\calO_6$ & $\calO_1$ & $\calO_3$ & $\calO_1$ & $\calO_1$ \\
   \hline
  2 & $\calO_1$ & $\calO_1$ & $\calO_4$ & $\calO_4$ & $\calO_1$ & $\calO_1$ & $\calO_1$ & $\calO_1$ \\
   \hline
  3 & $\calO_5$ & $\calO_5$ & $\calO_6$ & $\calO_2$ & $\calO_3$ & $\calO_1$ & $\calO_1$ & $\calO_1$ \\
   \hline
  4 & $\calO_1$ & $\calO_1$ & $\calO_1$ & $\calO_1$ & $\calO_1$ & $\calO_1$ & $\calO_1$ & $\calO_1$ \\
   \hline
  5 & $\calO_2$ & $\calO_6$ & $\calO_2$ & $\calO_1$ & $\calO_1$ & $\calO_1$ & $\calO_1$ & $\calO_1$ \\
  \hline
  6 & $\calO_4$ & $\calO_4$ & $\calO_1$ & $\calO_1$ & $\calO_1$ & $\calO_1$ & $\calO_1$ & $\calO_1$ \\
  \hline
  7 & $\calO_6$ & $\calO_1$ & $\calO_1$ & $\calO_1$ & $\calO_1$ & $\calO_1$ & $\calO_1$ & $\calO_1$ \\ 
  \hline
  \end{tabular}
\caption{3-qubit $\text{SLOCC}$ orbits of periodic states depending on their shift $l$ and period $r$.}\label{periodic3res}
 \end{center}
\end{table}

\begin{table}[!h]
 \begin{center}
  \begin{tabular}{|l||c|c|c|c|c|c|c|c|}
  \hline
   \backslashbox{$r$}{$l$} & 0 & 1 & 2 & 3 & 4 & 5 & 6 & 7    \\
   \hline
   \hline
  1 & $\calO_1$ & $\calO_6$ & $\calO_6$ & $\calO_6$ & $\calO_6$ & $\calO_6$ & $\calO_6$ & $\calO_1$ \\
   \hline
  2 & $\calO_1$ & $\calO_1$ & $\calO_3$ & $\calO_3$ & $\calO_3$ & $\calO_3$ & $\calO_1$ & $\calO_1$ \\
   \hline
  3 & $\calO_6$ & $\calO_6$ & $\calO_6$ & $\calO_6$ & $\calO_6$ & $\calO_1$ & $\calO_1$ & $\calO_1$ \\
   \hline
  4 & $\calO_1$ & $\calO_1$ & $\calO_1$ & $\calO_1$ & $\calO_1$ & $\calO_1$ & $\calO_1$ & $\calO_1$ \\
   \hline
  5 & $\calO_6$ & $\calO_6$ & $\calO_6$ & $\calO_1$ & $\calO_1$ & $\calO_1$ & $\calO_1$ & $\calO_1$ \\
  \hline
  6 & $\calO_3$ & $\calO_3$ & $\calO_1$ & $\calO_1$ & $\calO_1$ & $\calO_1$ & $\calO_1$ & $\calO_1$ \\
  \hline
  7 & $\calO_6$ & $\calO_1$ & $\calO_1$ & $\calO_1$ & $\calO_1$ & $\calO_1$ & $\calO_1$ & $\calO_1$ \\ 
  \hline
  \end{tabular}
\caption{3-qubit $\text{SLOCC}$ orbits of the resulting states after applying the QFT on periodic states depending on their shift $l$ and period $r$.}\label{periodic3resAFTERqft}
 \end{center}
\end{table}

\section{Algorithm}\label{appendixB}

\begin{algorithm}[H]
\caption{VerstraeteType}\label{algo2}
\begin{algorithmic}
\Require $Y$ an array of size 16, the 4-qubit state
\Ensure The Verstraete \textit{et al.} type of $Y$

~

\State Hess1 $\leftarrow Hess(\mathcal{Q}_1)$
\State Hess2 $\leftarrow Hess(\mathcal{Q}_2)$
\State Hess3 $\leftarrow Hess(\mathcal{Q}_3)$
\State T1 $\leftarrow T(\mathcal{Q}_1)$
\State T2 $\leftarrow T(\mathcal{Q}_2)$
\State T3 $\leftarrow T(\mathcal{Q}_3)$

~

\State \Comment{If the input form belongs to the nullcone}
\If{isInNullcone($Y$)}
\State \textbf{return} NilpotentType($Y$)
\EndIf

~

\State \Comment{The three quartics have at least a zero root}
\If{$L=0$ and $M=0$}

	\State \Comment{All the roots are simple}
	\If{$D_{xy}\neq 0$ and $Hyper \neq 0$} 
		\State \textbf{return} $G_{abc0}$
	
	~
	
	\State \Comment{All the zero roots are simple and there is a nonzero double root}
	\ElsIf{$D_{xy}\neq 0$ and $Hyper = 0$}
		\State vectCov $\leftarrow [\mathcal{L}]$
		\State eval $\leftarrow$ evaluate(vectCov,$Y$)
		\If{eval $= [0]$}
			\State \textbf{return} $G_{aa(-2a)0}$
		\Else
			\State \textbf{return} $L_{0b(\frac{b}{2})_2}$
		\EndIf
		
	~
	
	\State \Comment{All the zero roots are double}
	\ElsIf{$D_{xy}= 0$ and $H \neq 0$}	
		\State vectCov $\leftarrow [\overline{\mathcal{G}}, \mathcal{G},\mathcal{H},\mathcal{L}]$
		\State eval $\leftarrow$ evaluate(vectCov,$Y$)
		\If{eval $= [0,0,0,0]$}
			\State \textbf{return} $G_{00cc}$
		\ElsIf{eval $= [0,1,1,0]$}
			\State \textbf{return} $L_{aa0_2}$
		\ElsIf{eval $= [0,0,1,0]$}
			\State \textbf{return} $L_{00c_2}$	
		\ElsIf{eval $= [1,1,1,0]$}
			\State \textbf{return} $L_{0_2b_2}$	
		\ElsIf{eval $= [1,1,1,1]$}
			\State \textbf{return} $L_{a_20_{3\oplus \overline{1}}}$	
		\EndIf
		
	\EndIf	

~

\State \Comment{Only one of the quartics $\mathcal{Q}_i$ has a zero root then}
\State \Comment{The quartic $\mathcal{Q}_1$ has a zero root}
\ElsIf{$L=0$ and $M\neq 0$}	

	\State \Comment{$\mathcal{Q}_1$ has only simple roots}	
	\If{$Hyper \neq 0$} 
		\State \textbf{return} $G_{abc0}$
	
\algstore{myalg}
\end{algorithmic}
\end{algorithm}

\begin{algorithm}[H]
\begin{algorithmic}
\algrestore{myalg}
	
	\State \Comment{$\mathcal{Q}_1$ has a double zero root and two simple roots}
	\ElsIf{$D_{xy}=H\cdot M$ and $H^2+4M \neq 0$}
		\State vectCov $\leftarrow [\mathcal{K}_3,\mathcal{L}]$
		\State eval $\leftarrow$ evaluate(vectCov,$Y$)
		\If{eval $= [0,0]$}
			\State \textbf{return} $G_{ab00}$
		\ElsIf{eval $= [1,0]$}
			\State \textbf{return} $L_{ab0_2}$
		\ElsIf{eval $= [1,1]$}
			\State \textbf{return} $L_{a_2b_2}$	
		\EndIf

~

	\State \Comment{$\mathcal{Q}_1$ has a double nonzero root and two simple roots}
	\ElsIf{$D_{xy}\neq H\cdot M$ and T1$ \neq 0$ and T2$ \neq 0$}
		\State vectCov $\leftarrow [\mathcal{L}]$
		\State eval $\leftarrow$ evaluate(vectCov,$Y$)
		\If{eval $= [0]$}
			\State \textbf{return} $G_{abb0}$
		\ElsIf{eval $= [1]$}
			\State \textbf{return} $L_{a0c_2}$
		\EndIf
			
	~
	
	\State \Comment{$\mathcal{Q}_1$ has a triple zero root and a simple root}
	\ElsIf{$D_{xy} = H\cdot M$ and $H^2+4M =0$ and Hess2$ = 0$}	
		\State vectCov $\leftarrow [\mathcal{C},\mathcal{D},\mathcal{K}_5,\mathcal{L}]$
		\State eval $\leftarrow$ evaluate(vectCov,$Y$)
		\If{eval $= [0,0,0,0]$}
			\State \textbf{return} $G_{a000}$
		\ElsIf{eval $= [1,0,0,0]$}
			\State \textbf{return} $L_{a00_2}$
		\ElsIf{eval $= [1,1,1,0]$}
			\State \textbf{return} $L_{0b_3}$	
		\ElsIf{eval $= [1,1,0,0]$}
			\State \textbf{return} $L_{a_2a_2}$	
		\ElsIf{eval $= [1,1,1,1]$}
			\State \textbf{return} $L_{a_4}$	
		\EndIf
	
	~
	
	\State \Comment{$\mathcal{Q}_1$ has a triple nonzero root}
	\ElsIf{$D_{xy} \neq H\cdot M$ and $I_2 =0$ and $Hyper = 0$}		
		\State vectCov $\leftarrow [\mathcal{D},\mathcal{L}]$
		\State eval $\leftarrow$ evaluate(vectCov,$Y$)
		\If{eval $= [0,0]$}
			\State \textbf{return} $G_{aaa0}$
		\ElsIf{eval $= [1,0]$}
			\State \textbf{return} $L_{0bb_2}$
		\ElsIf{eval $= [1,1]$}
			\State \textbf{return} $L_{a0_3}$		
		\EndIf
	\EndIf
	
~

~

\State \Comment{The quartic $\mathcal{Q}_2$ has a zero root}
\ElsIf{$M=0$ and $L\neq 0$}	

	\State \Comment{$\mathcal{Q}_2$ has only simple roots}	
	\If{$Hyper \neq 0$} 
		\State \textbf{return} $G_{abc0}$
	
	~
	
	\State \Comment{$\mathcal{Q}_2$ has a double zero root and two simple roots}
	\ElsIf{$D_{xy}=0$ and $H^2\neq 4L $}
		\State vectCov $\leftarrow [\mathcal{K}_3,\mathcal{L}]$
		\State eval $\leftarrow$ evaluate(vectCov,$Y$)
		\If{eval $= [0,0]$}
			\State \textbf{return} $G_{ab00}$
		\ElsIf{eval $= [1,0]$}
			\State \textbf{return} $L_{ab0_2}$
		\ElsIf{eval $= [1,1]$}
			\State \textbf{return} $L_{a_2b_2}$	
		\EndIf

	\algstore{myalg2}
\end{algorithmic}
\end{algorithm}

\begin{algorithm}[H]
\begin{algorithmic}
\algrestore{myalg2}
								
	\State \Comment{$\mathcal{Q}_2$ has a double nonzero root and two simple roots}
	\ElsIf{$D_{xy}\neq 0$ and T1$ \neq 0$ and T2$ \neq 0$}
		\State vectCov $\leftarrow [\mathcal{L}]$
		\State eval $\leftarrow$ evaluate(vectCov,$Y$)
		\If{eval $= [0]$}
			\State \textbf{return} $G_{abb0}$
		\ElsIf{eval $= [1]$}
			\State \textbf{return} $L_{a0c_2}$
		\EndIf
			
~
	
	\State \Comment{$\mathcal{Q}_2$ has a triple zero root and a simple root}
	\ElsIf{$D_{xy} = 0$ and $H^2 = 4L$}	
		\State vectCov $\leftarrow [\mathcal{C},\mathcal{D},\mathcal{K}_5,\mathcal{L}]$
		\State eval $\leftarrow$ evaluate(vectCov,$Y$)
		\If{eval $= [0,0,0,0]$}
			\State \textbf{return} $G_{a000}$
		\ElsIf{eval $= [1,0,0,0]$}
			\State \textbf{return} $L_{a00_2}$
		\ElsIf{eval $= [1,1,1,0]$}
			\State \textbf{return} $L_{0b_3}$	
		\ElsIf{eval $= [1,1,0,0]$}
			\State \textbf{return} $L_{a_2a_2}$	
		\ElsIf{eval $= [1,1,1,1]$}
			\State \textbf{return} $L_{a_4}$	
		\EndIf
	
	~
	
	\State \Comment{$\mathcal{Q}_2$ has a triple nonzero root}
	\ElsIf{$D_{xy} \neq 0$ and $I_2 =0$ and $Hyper = 0$}		
		\State vectCov $\leftarrow [\mathcal{D},\mathcal{L}]$
		\State eval $\leftarrow$ evaluate(vectCov,$Y$)
		\If{eval $= [0,0]$}
			\State \textbf{return} $G_{aaa0}$
		\ElsIf{eval $= [1,0]$}
			\State \textbf{return} $L_{0bb_2}$
		\ElsIf{eval $= [1,1]$}
			\State \textbf{return} $L_{a0_3}$		
		\EndIf
	\EndIf			
				
~

~

\State \Comment{The quartic $\mathcal{Q}_3$ has a zero root}
\ElsIf{$N=0$ and $L\neq 0$ and $M \neq 0$}	

	\State \Comment{$\mathcal{Q}_3$ has only simple roots}	
	\If{$Hyper \neq 0$} 
		\State \textbf{return} $G_{abc0}$
	
	~
	
	\State \Comment{$\mathcal{Q}_3$ has a double zero root and two simple roots}
	\ElsIf{$D_{xy}=0$ and $H^2\neq 4M $}
		\State vectCov $\leftarrow [\mathcal{K}_3,\mathcal{L}]$
		\State eval $\leftarrow$ evaluate(vectCov,$Y$)
		\If{eval $= [0,0]$}
			\State \textbf{return} $G_{ab00}$
		\ElsIf{eval $= [1,0]$}
			\State \textbf{return} $L_{ab0_2}$
		\ElsIf{eval $= [1,1]$}
			\State \textbf{return} $L_{a_2b_2}$	
		\EndIf

	 ~
								
	\State \Comment{$\mathcal{Q}_3$ has a double nonzero root and two simple roots}
	\ElsIf{$D_{xy}\neq 0$ and T1$ \neq 0$ and T2$ \neq 0$}
		\State vectCov $\leftarrow [\mathcal{L}]$
		\State eval $\leftarrow$ evaluate(vectCov,$Y$)
		\If{eval $= [0]$}
			\State \textbf{return} $G_{abb0}$
		\ElsIf{eval $= [1]$}
			\State \textbf{return} $L_{a0c_2}$
		\EndIf
			
\algstore{myalg3}
\end{algorithmic}
\end{algorithm}

\begin{algorithm}[H]
\begin{algorithmic}	
\algrestore{myalg3}
 
	\State \Comment{$\mathcal{Q}_3$ has a triple zero root and a simple root}
	\ElsIf{$D_{xy} = 0$ and $H^2 = 4M$}	
		\State vectCov $\leftarrow [\mathcal{C},\mathcal{D},\mathcal{K}_5,\mathcal{L}]$
		\State eval $\leftarrow$ evaluate(vectCov,$Y$)
		\If{eval $= [0,0,0,0]$}
			\State \textbf{return} $G_{a000}$
		\ElsIf{eval $= [1,0,0,0]$}
			\State \textbf{return} $L_{a00_2}$
		\ElsIf{eval $= [1,1,1,0]$}
			\State \textbf{return} $L_{0b_3}$	
		\ElsIf{eval $= [1,1,0,0]$}
			\State \textbf{return} $L_{a_2a_2}$	
		\ElsIf{eval $= [1,1,1,1]$}
			\State \textbf{return} $L_{a_4}$	
		\EndIf
	
~
	
	\State \Comment{$\mathcal{Q}_3$ has a triple nonzero root}
	\ElsIf{$D_{xy} \neq 0$ and $I_2 =0$ and $Hyper = 0$}		
		\State vectCov $\leftarrow [\mathcal{D},\mathcal{L}]$
		\State eval $\leftarrow$ evaluate(vectCov,$Y$)
		\If{eval $= [0,0]$}
			\State \textbf{return} $G_{aaa0}$
		\ElsIf{eval $= [1,0]$}
			\State \textbf{return} $L_{0bb_2}$
		\ElsIf{eval $= [1,1]$}
			\State \textbf{return} $L_{a0_3}$		
		\EndIf
	\EndIf		

~

~

\State \Comment{All the quartics have only nonzero roots}
\Else

	\State \Comment{All the roots are simple}	
	\If{$Hyper \neq 0$} 
		\State \textbf{return} $G_{abcd}$
	
	~
	
	\State \Comment{Each quartic has a double root and two simple roots}
	\ElsIf{T1 $\neq 0$ and T2 $\neq 0$ and T3 $\neq 0$ and $I_2 \neq 0$ and $I_3 \neq 0$}
		\State vectCov $\leftarrow [\mathcal{L}]$
		\State eval $\leftarrow$ evaluate(vectCov,$Y$)
		\If{eval $= [0]$}
			\State \textbf{return} $G_{abcc}$
		\ElsIf{eval $= [1]$}
			\State \textbf{return} $L_{abc_2}$
		\EndIf
		
	~
	
	\State \Comment{Each quartic has a single simple root and a triple root}
	\ElsIf{$I_2 \neq 0$ and $I_3 \neq 0$ and Hess1 $\neq 0$}
		\State vectCov $\leftarrow [\mathcal{K}_5,\mathcal{L}]$
		\State eval $\leftarrow$ evaluate(vectCov,$Y$)
		\If{eval $= [0,0]$}
			\State \textbf{return} $G_{abbb}$
		\ElsIf{eval $= [1,0]$}
			\State \textbf{return} $L_{abb_2}$
		\ElsIf{eval $= [1,1]$}
			\State \textbf{return} $L_{ab_3}$
		\EndIf

	~
		
	\EndIf
	
	~

\EndIf

\end{algorithmic}
\end{algorithm}

\end{document}